\documentclass[11pt]{article}
\usepackage{times}
\usepackage{amssymb}
\usepackage{amsmath}
\usepackage{enumerate,xspace,multirow,amsthm}
\usepackage{graphicx}
\usepackage[shortlabels]{enumitem}
\usepackage{comment}
\usepackage{calc}
\usepackage{xr}
\usepackage{chngcntr}
\usepackage{tkz-graph}
\usepackage{tikz}

\usepackage{pgfplots}
\usetikzlibrary{arrows,shapes,decorations.pathmorphing,calc}
\usetikzlibrary{positioning}
\pgfdeclarelayer{background}
\pgfsetlayers{background,main}

\newcounter{xx}
\newcounter{yy}
\newlength{\rad}
\newtheorem{theorem}{Theorem}[section] 
\newtheorem{lemma}[theorem]{Lemma}
\newtheorem{claim}[theorem]{Claim}

{\theoremstyle{remark} }
{\theoremstyle{definition} 
  \newtheorem{definition}[theorem]{Definition}}

\newcommand{\bg}[1]{\medskip\noindent{\it #1}}

\newenvironment{proofof}[1]{\begin{proof}[Proof of {#1}]}{\end{proof}}

\newcommand{\group}[1]{\mathbb{Z}_{#1}}
\newcommand{\grp}{\group{2}}

\newcommand{\mc}{\mathcal}
\newcommand{\gap}{\enspace}
\newcommand{\gaptwo}{\gap \gap}
\newcommand{\gapfive}{\gaptwo \gaptwo \gap}
\newcommand{\gapten}{\gapfive \gapfive}

\newcommand{\T}{\ensuremath{\mc T}}
\newcommand{\hT}{\ensuremath{\widehat \T}}
\newcommand{\hG}{\ensuremath{\widehat G}}

\newcommand{\Tu}{\ensuremath{\T_{uu}}}
\newcommand{\Tuv}{\ensuremath{\T_{uv}}}
\newcommand{\Tv}{\ensuremath{\T_{vv}}}
\newcommand{\Pc}{\ensuremath{\mc P}}
\newcommand{\Zc}{\ensuremath{\mc Z}}

\newcommand{\C}{\ensuremath{\mc C}}
\newcommand{\K}{\ensuremath{\mc K}}
\newcommand{\Sc}{\ensuremath{\mc S}}

\newcommand{\ceil}[1]{\ensuremath{\left\lceil#1\right\rceil}}
\newcommand{\floor}[1]{\ensuremath{\left\lfloor#1\right\rfloor}}

\newcommand{\gm}{\ensuremath{\gamma}}
\newcommand{\ld}{\ensuremath{\lambda}}

\newcommand{\dt}{\ensuremath{\delta}}

\newcommand{\Gm}{\ensuremath{\Gamma}}
\newcommand{\sg}{\ensuremath{\sigma}}
\newcommand{\Sg}{\ensuremath{\Sigma}}

\newcommand{\uv}{\ensuremath{\{u,v\}}}

\newcommand{\sm}{\ensuremath{\setminus}}
\newcommand{\es}{\ensuremath{\emptyset}}
\newcommand{\sse}{\subseteq}

\newcommand{\bQ}{\ensuremath{\overline{Q}}}
\newcommand{\tA}{\ensuremath{\widetilde A}}
\newcommand{\tH}{\ensuremath{\widetilde H}}
\newcommand{\tgm}{\ensuremath{\widetilde\gamma}}
\DeclareMathOperator{\comp}{comp}

\title{Min-Max Theorems for Packing and Covering Odd $(u,v)$-trails%
\thanks{A preliminary version~\cite{IbrahimpurS17} appeared in the Proceedings of the
  Internaitonal Conference on Integer Programming and Combinatorial Optimization (IPCO),
  2017.}} 

\author{
         Sharat Ibrahimpur\thanks{{\tt \{sharat.ibrahimpur, cswamy\}@uwaterloo.ca}. 
         Dept. of Combinatorics and Optimization, Univ. Waterloo, Waterloo, ON N2L 3G1.
         Supported in part by NSERC grant 327620-09 and an NSERC Discovery Accelerator
         Supplement award.}   
\and
\addtocounter{footnote}{-1}
         Chaitanya Swamy\footnotemark
}

\date{}

\begin{document}

\maketitle

\begin{abstract}
We investigate the problem of packing and covering odd $(u,v)$-trails in a
graph. A $(u,v)$-trail is a $(u,v)$-walk that is allowed to have repeated vertices but no  
repeated edges. We call a trail \emph{odd} if the number of edges in the trail is
odd. 
Let $\nu(u,v)$ denote the maximum number of edge-disjoint odd $(u,v)$-trails, 
and $\tau(u,v)$ denote the minimum size of an edge-set that intersects every odd
$(u,v)$-trail. 

We prove that $\tau(u,v)\leq 2\nu(u,v)+1$. 
Our result is {\em tight}---there are examples showing that $\tau(u,v)=2\nu(u,v)+1$---%
and substantially improves upon the bound of $8$ obtained in~\cite{ChurchleyMW16} for
$\tau(u,v)/\nu(u,v)$. 
Our proof also yields a polynomial-time algorithm for finding a cover and a collection of 
trails satisfying the above bounds.  

Our proof is simple and has two main ingredients. We show that 
(loosely speaking) the problem can be reduced to the problem of packing and covering odd
$(\uv,\uv)$-trails losing a factor of 2 (either in the number of trails found, or the size
of the cover). Complementing this, we show that the odd-$(\uv,\uv)$-trail packing and
covering problems can be tackled by exploiting a powerful min-max result
of~\cite{ChudnovskyGGGLS06} for packing vertex-disjoint nonzero $A$-paths in group-labeled
graphs.  
\end{abstract}

\section{Introduction}
Min-max theorems are a classical and central theme in combinatorics and combinatorial
optimization, with many such results arising from the study of packing and covering
problems. For instance, {\em Menger's theorem}~\cite{Menger27} gives a {\em tight} min-max   
relationship for packing and covering edge-disjoint (or internally vertex-disjoint)
$(u,v)$-paths: the maximum number of edge-disjoint (or internally vertex-disjoint)
$(u,v)$-paths (i.e., {\em packing number}) is equal to the minimum number of edges 
(or vertices) needed to cover all $u$-$v$ paths (i.e., {\em covering number}); the
celebrated {\em max-flow min-cut theorem} generalizes this result to arbitrary
edge-capacitated graphs.  
Another well-known example is the {\em Lucchesi-Younger theorem}~\cite{LucchesiY78}, which
shows that the maximum number of edge-disjoint directed cuts equals the minimum-size of an
arc-set that intersects every directed cut. 

Motivated by Menger's theorem, it is natural to ask whether similar (tight or approximate)
min-max theorems hold for other variants 
of path-packing and path-covering problems. Questions of this flavor have attracted 
a great deal of attention. Perhaps the most prominent results known of this type are
Mader's min-max theorems for packing vertex-disjoint
$\Sc$-paths~\cite{Mader78,Schrijver01},   
which generalize both the {\em Tutte-Berge formula} and Menger's theorem, and further 
far-reaching generalizations of this due to Chudnovsky et al.~\cite{ChudnovskyGGGLS06} and
Pap~\cite{Pap07} regarding packing vertex-disjoint non-zero paths and non-returning
$A$-paths in group-labeled graphs and permutation-labeled graphs respectively.
 
We consider a different variant of the $(u,v)$-path 
packing and covering problems, wherein we impose {\em parity constraints} on the
paths. Such constraints naturally arise in the study of multicommodity-flow problem, which
can be phrased in terms of packing odd circuits in a signed graph, and consequently, such
odd-circuit packing and covering problems have been widely investigated (see,
e.g.,~\cite{Schrijver03}, Chapter~75).   
Focusing on $(u,v)$-paths, a natural variant that arises involves packing and
covering {\em odd} $(u,v)$-paths, where a $(u,v)$-path is odd if it contains an odd number
of edges. 
However, there are simple examples~\cite{ChurchleyMW16} showing an unbounded gap between
the packing and covering numbers in this setting. 

In light of this, following~\cite{ChurchleyMW16}, we investigate the min-max relationship
for packing and covering odd $(u,v)$-trails. An {\em odd $(u,v)$-trail} is a $(u,v)$-walk
with {\em no repeated edges} and an odd number of edges. 
Churchley et al.~\cite{ChurchleyMW16} seem to have been the first to consider this
problem. They showed 
that the (worst-case) ratio between the covering and packing numbers for odd
$(u,v)$-trails is at most $8$---which is in stark contrast with the setting of odd $(u,v)$
paths, where the ratio is unbounded---and at least $2$, 
so there is no tight min-max theorem like Menger's theorem. 
They motivate the study of odd $(u,v)$-trails 
from the perspective of studying {\em totally-odd immersions}. In particular, determining
if a graph $G$ has $k$ edge-disjoint odd $(u,v)$-trails is equivalent to deciding if the
2-vertex graph with $k$ parallel edges has a totally-odd immersion into $G$. 

\vspace{-1ex}
\paragraph{Our results.}
We prove a {\em tight bound} on the ratio of the covering and packing numbers for
odd $(u,v)$-trails, which also substantially improves the bound of
$8$ shown in~\cite{ChurchleyMW16} for this covering-vs-packing ratio.%
\footnote{This bound was later improved to $5$~\cite{Churchley16,Ibrahimpur16}.
We build upon some of the ideas in~\cite{Ibrahimpur16}.} 
Let $\nu(u,v)$ and $\tau(u,v)$ denote respectively the packing and covering numbers for
odd $(u,v)$-trails. 
Our main result (Theorem~\ref{thm:2k}) establishes that $\tau(u,v)\leq 2\nu(u,v)+1$.
Furthermore, we obtain in polynomial time a certificate establishing that $\tau(u,v)\leq
2\nu(u,v)+1$. This is because we show that, for any integer $k\geq 0$, we 
can compute in polynomial time, a collection of $k$ edge-disjoint odd
$(u,v)$-trails, or an odd-$(u,v)$-trail cover of size at most $2k-1$.
As mentioned earlier, there are examples showing $\tau(u,v)=2\nu(u,v)+1$ (see
Fig.~\ref{2gap}), so our result {\em settles} the question of obtaining worst-case
bounds for the $\tau(u,v)/\nu(u,v)$ ratio.  

Notably, our proof is also simple, and noticeably simpler than (and different from) the
one in~\cite{ChurchleyMW16}. 
We remark that the proof in~\cite{ChurchleyMW16} constructs covers of a
certain form; in Appendix~\ref{append-lbound}, we prove a {\em lower bound} showing that
such covers cannot yield a bound better than $3$ on the covering-vs-packing ratio. 

\vspace{-1ex}
\paragraph{Our techniques.}
We focus on showing that for any $k$, we can obtain either $k$
edge-disjoint odd $(u,v)$-trails or a cover of size at most $2k-1$.
This follows from two other auxiliary results which are potentially of
independent interest. 

Our key insight is that one can {\em decouple} the requirements of
parity and $u$-$v$ connectivity when constructing odd $(u,v)$-trails. More precisely, we
show that if we have a collection of $k$ edge-disjoint odd $(\uv,\uv)$-trails, that is,
odd trails that start and end at a vertex of $\uv$, and the $u$-$v$ edge
connectivity, denoted $\ld(u,v)$, is at least $2k$, then we can obtain $k$ edge-disjoint
odd $(u,v)$-trails (Theorem~\ref{thm:contacts}). Notice that if $\ld(u,v)<2k$, then a min
$u$-$v$ cut yields a cover of the desired size. So the upshot of
Theorem~\ref{thm:contacts} is that it reduces our task to the {\em relaxed} problem of
finding $k$ edge-disjoint odd $(\uv,\uv)$-trails. 
The proof of Theorem~\ref{thm:contacts}
relies on elementary arguments (see Section~\ref{sec:contacts}). We show that given a
fixed collection of $2k$ edge-disjoint $(u,v)$-paths, we can always modify our collection
of edge-disjoint trails so as to make progress by decreasing the number of contacts that
the paths make with the trails and/or by increasing the number of odd $(u,v)$ trails in
the collection. Repeating this process a small number of times thus yields the $k$
edge-disjoint odd $(u,v)$-trails.  

Complementing Theorem~\ref{thm:contacts} we prove that we can either obtain $k$
edge-disjoint $(\uv,\uv)$-trails, or find an odd-$(\uv,\uv)$-trail cover (which is also an   
odd-$(u,v)$-trail cover) of size at most $2k-2$ (Theorem~\ref{thm:sstrails}). This proof
relies on a powerful result of~\cite{ChudnovskyGGGLS06} about {\em packing and covering nonzero
$A$-paths in group-labeled graphs} (see Section~\ref{sec:sstrails}, which defines these
concepts precisely).  
The idea here is that~\cite{ChudnovskyGGGLS06} show that one can obtain either $k$ vertex-disjoint
nonzero $A$-paths or a set of at most $2k-2$ vertices intersecting all nonzero
$A$-paths, and this can be done in polytime~\cite{ChudnovskyCG08,Geelen16} (see
also~\cite{Pap08}).  
This is the same type of result that we seek, except that 
we care about edge-disjoint trails, as opposed to vertex-disjoint paths. However, by
moving to a suitable gadget graph (essentially the line graph) where we replace each
vertex by a clique, we can encode trails as paths, and edge-disjointness is captured by
vertex-disjointness. Applying the result in~\cite{ChudnovskyGGGLS06} then yields
Theorem~\ref{thm:sstrails}.  

\vspace{-1ex}
\paragraph{Related work.}
Churchley et al.~\cite{ChurchleyMW16} initiated the study of min-max 
theorems for packing and covering odd $(u,v)$-trails. They cite the question of
totally-odd immersions as motivation for their work. We say that a graph $H$ has an
{\em immersion}~\cite{RobertsonS10} into another graph $G$, if one can map $V_H$
bijectively to some $U\sse V(G)$, and $E_H$ to edge-disjoint trails connecting the
corresponding vertices in $U$.  
(As noted by~\cite{ChurchleyMW16}, trails are more natural objects than paths in 
the context of reversing an edge-splitting-off operation, as this, in general, creates
trails.)  
An immersion is called {\em totally odd} if all trails are of odd length. 
As noted earlier, the question of deciding if a graph $G$ has $k$ edge-disjoint odd
$(u,v)$-trails can be restated as determining if the 2-vertex graph with $k$ parallel
edges has a totally-odd immersion into $G$.  

In an interesting contrast to the unbounded gap between the covering and packing numbers
for odd $(u,v)$-paths, \cite{SchrijverS94} 
showed that the covering number is at most twice the {\em fractional packing number}
(which is the optimal value of the natural odd-$(u,v)$-path-packing LP).

The notions of odd paths and trails can be generalized and abstracted in two ways. The
first involves {\em signed graphs}~\cite{Zaslavsky98}, and there are various results on packing
odd {\em circuits} in signed graphs, which are closely related to multicommodity flows
(see \cite{Schrijver03}, Chapter~75). The second involves 
{\em group-labeled graphs}, for which~\cite{ChudnovskyGGGLS06,ChudnovskyCG08} present
strong min-max theorems and algorithms for packing and covering vertex-disjoint nonzero
$A$-paths. 
Pap further generalized the latter results to the setting of packing vertex-disjoint
non-returning $A$-paths in {\em permutation-labeled graphs}. He obtained both a min-max
theorem for the packing problem~\cite{Pap07} (which is analogous to the min-max theorem
in~\cite{ChudnovskyGGGLS06}), and devised an algorithm for computing a maximum-cardinality
packing~\cite{Pap08}.

\section{Preliminaries and notation} \label{sec:prelims}
Let $G = (V,E)$ be an undirected graph. 
For $X\sse V$, we use $E(X)$ to denote the set of edges having both endpoints in $X$ and
$\delta(X)$ to denote set of edges with exactly one endpoint in $X$. For disjoint 
$X,Y\sse V$, we use $E(X,Y)$ to denote the set of edges with one end in $X$ and one end in  
$Y$. 

A {\em $(p,q)$-walk} is a sequence $(x_0,e_1,x_1,e_2,x_2,\dots,e_r,x_r)$, where
$x_0,\ldots,x_r\in V$ with $x_0=p$, $x_r=q$, and $e_i$ is an edge with ends $x_{i-1}$,
$x_i$ for all $i=1,\ldots,r$. The vertices $x_1,\ldots,x_{r-1}$
are called the {\em internal vertices} of this walk. We say that such a $(p,q)$-walk is a:  
\begin{itemize}[nosep]
\item {\em $(p,q)$-path}, if either $r>0$ and all the $x_i$s 
are distinct (so $p\neq q$), or $r=0$, 
which we call a {\em trivial path}; 
\item {\em $(p,q)$-trail} if all the $e_i$s are distinct (we could have $p=q$).
\end{itemize}
Thus, a $(p,q)$-trail is a $(p,q)$-walk that is allowed to have repeated vertices but  
{\em no repeated edges}. 
Given vertex-sets $A,B \subseteq V$, we say that a trail is an
$(A,B)$-trail to denote that 
it is a $(p,q)$-trail for some $p \in A, q \in B$. 
A $(p,q)$-trail is called {\em odd} (respectively, {\em even}) if
it has an odd (respectively, even) number of edges.

\begin{definition}
Let $G = (V,E)$ be a graph, and $u,v\in V$ (we could have $u=v$). 
\begin{enumerate}[(a), nosep]
\item The \emph{packing number for odd $(u,v)$-trails}, denoted $\nu(u,v;G)$, is the
maximum number of edge-disjoint odd $(u,v)$-trails in $G$. 
\item We call a subset of edges $C$ an {\em odd $(u,v)$-trail cover} of $G$ if it
intersects every odd $(u,v)$-trail in $G$. The 
{\em covering number for odd $(u,v)$-trails},
denoted $\tau(u,v;G)$, is the minimum size of an odd $(u,v)$-trail cover of $G$. 
\end{enumerate}
We drop the argument $G$ when it is clear from the context.
\end{definition}

For any two distinct vertices $x,y$ of $G$, we denote the size of a minimum
$(x,y)$-cut in $G$ by $\lambda(x,y;G)$, and drop $G$ when it is clear from the context. 
By the max-flow min-cut (or Menger's) theorem, $\ld(x,y;G)$ is also the maximum number of 
edge-disjoint $(x,y)$-paths in $G$.

\section{Main results and proof overview} \label{sec:main}

Our main result is the following {\em tight} approximate min-max theorem relating
the packing and covering numbers for odd $(u,v)$ trails. 

\begin{theorem} \label{thm:2k}
Let $G=(V,E)$ be an undirected graph, and $u, v\in V$. 
For any nonnegative integer $k$, we can obtain in polynomial time, either:
\begin{enumerate}[topsep=0.5ex, itemsep=0.5ex]
\item $k$ edge-disjoint odd $(u,v)$-trails in $G$, or
\item an odd $(u,v)$-trail cover of $G$ of size at most $2k-1$.
\end{enumerate}
Hence, we have $\tau(u,v;G) \leq 2 \cdot \nu(u,v;G) + 1$.
\end{theorem}

Theorem~\ref{thm:2k} is tight, as can be seen from Fig.~\ref{2gap}; we show in
Appendix~\ref{append-2k} that $\nu(u,v;G)=k$ and $\tau(u,v;G)=2k+1$ for this instance.
(The fact that Theorem~\ref{thm:2k} is tight was communicated to us by~\cite{Churchley16},
who provided a different tight example.) 

\begin{figure}[h!]
\centering
\begin{tikzpicture}

\setcounter{xx}{0}
\setcounter{yy}{0}
\coordinate (centre) at (\value{xx},\value{yy});
\node (dummy) at (centre) {};

\coordinate (u) at ([xshift=-5.5cm]centre);

\coordinate (aa) at ([xshift=2.5cm,yshift=2.2cm]u);
\coordinate (ba) at ([yshift=-1.4cm]aa);
\coordinate (ca) at ([xshift=1cm,yshift=-0.7cm]aa);
\coordinate (da) at ([xshift=1cm,yshift=0.7cm]ca);
\coordinate (ea) at ([yshift=-1.4cm]da);
\coordinate (fa) at ([xshift=1cm,yshift=-0.7cm]da);
\coordinate (ga) at ([xshift=1cm,yshift=0.7cm]fa);
\coordinate (ha) at ([yshift=-1.4cm]ga);

\coordinate (v) at ([xshift=2.5cm,yshift=-2.2cm]ga);
\coordinate (w) at ([yshift=1.5cm]da);

\draw [fill=black] (u) circle (\rad) node[left] {$u$};
\draw [fill=black] (v) circle (\rad) node[right] {$v$};
\draw [fill=black] (w) circle (\rad) node[above] {$w$};

\draw [fill=black] (aa) circle (\rad) node[above] {$a_1$};
\draw [fill=black] (ba) circle (\rad) node[below] {$b_1$};
\draw [fill=black] (ca) circle (\rad) node[above] {$c_1$};
\draw [fill=black] (da) circle (\rad) node[above] {$d_1$};
\draw [fill=black] (ea) circle (\rad) node[below] {$e_1$};
\draw [fill=black] (fa) circle (\rad) node[above] {$f_1$};
\draw [fill=black] (ga) circle (\rad) node[above] {$g_1$};
\draw [fill=black] (ha) circle (\rad) node[below] {$h_1$};

\draw (u) -- (aa);
\draw (u) -- (ba);
\draw (aa) -- (ba);
\draw (aa) -- (ca);
\draw (ba) -- (ca);
\draw (ca) -- (da);
\draw (ca) -- (ea);
\draw (da) -- (ea);
\draw (da) -- (fa);
\draw (ea) -- (fa);
\draw (fa) -- (ga);
\draw (fa) -- (ha);
\draw (ga) -- (v);
\draw (ha) -- (v);
\draw (ga) -- (ha);

\coordinate (ab) at ([xshift=2.5cm,yshift=-0.8cm]u);
\coordinate (bb) at ([yshift=-1.4cm]ab);
\coordinate (cb) at ([xshift=1cm,yshift=-0.7cm]ab);
\coordinate (db) at ([xshift=1cm,yshift=0.7cm]cb);
\coordinate (eb) at ([yshift=-1.4cm]db);
\coordinate (fb) at ([xshift=1cm,yshift=-0.7cm]db);
\coordinate (gb) at ([xshift=1cm,yshift=0.7cm]fb);
\coordinate (hb) at ([yshift=-1.4cm]gb);

\draw [fill=black] (ab) circle (\rad) node[above] {$a_k$};
\draw [fill=black] (bb) circle (\rad) node[below] {$b_k$};
\draw [fill=black] (cb) circle (\rad) node[above] {$c_k$};
\draw [fill=black] (db) circle (\rad) node[above] {$d_k$};
\draw [fill=black] (eb) circle (\rad) node[below] {$e_k$};
\draw [fill=black] (fb) circle (\rad) node[above] {$f_k$};
\draw [fill=black] (gb) circle (\rad) node[above] {$g_k$};
\draw [fill=black] (hb) circle (\rad) node[below] {$h_k$};

\draw (u) -- (ab);
\draw (u) -- (bb);
\draw (ab) -- (bb);
\draw (ab) -- (cb);
\draw (bb) -- (cb);
\draw (cb) -- (db);
\draw (cb) -- (eb);
\draw (db) -- (eb);
\draw (db) -- (fb);
\draw (eb) -- (fb);
\draw (fb) -- (gb);
\draw (fb) -- (hb);
\draw (gb) -- (v);
\draw (hb) -- (v);
\draw (gb) -- (hb);

\node[font=\bf] at ([yshift=-0.7cm]ea) {\vdots};

\draw (u) edge[bend left] (w);
\draw (v) edge[bend right] (w);

\end{tikzpicture}

\caption{Graph with $\nu(u,v)=k$, $\tau(u,v)=2k+1$.} \label{fig:smallgap} \label{2gap}
\end{figure}
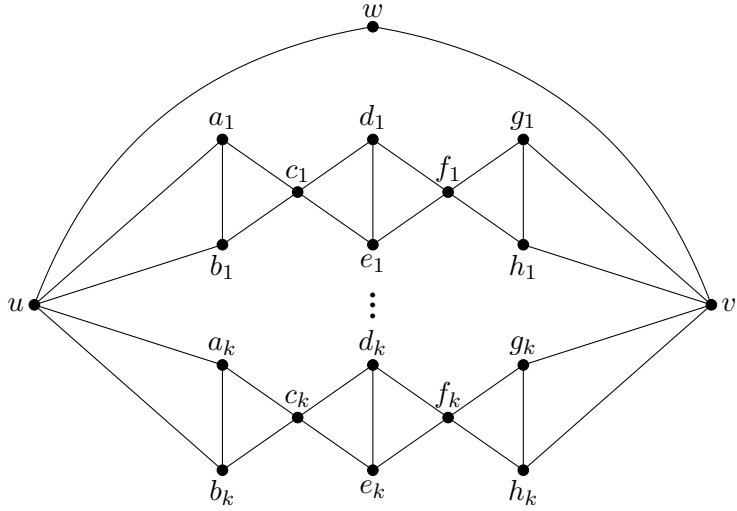

\begin{theorem} \label{thm:sstrails}
Let $G=(V,E)$ be an undirected graph and $s\in V$. 
For any nonnegative integer $k$, we can obtain in polynomial time: 
\begin{enumerate}[topsep=0.25ex, itemsep=0.25ex]
\item $k$ edge-disjoint odd $(s,s)$-trails in $G$, or
\item an odd $(s,s)$-trail cover of $G$ of size at most $2k-2$.
\end{enumerate}
\end{theorem}

Theorem~\ref{thm:2k} follows readily from 
the following two results. 

\begin{theorem} \label{thm:contacts}
Let $G=(V,E)$ be an undirected graph, and $u,v\in V$ with $u\neq v$. 
Let $\hT$ be a collection of edge-disjoint odd $(\uv,\uv)$-trails in $G$. If
$\lambda(u,v)\geq 2 \cdot |\hT|$, then we can obtain in polytime $|\hT|$ edge-disjoint odd 
$(u,v)$-trails in $G$.   
\end{theorem}

\begin{proofof}{Theorem~\ref{thm:2k}}
If $u=v$, then Theorem~\ref{thm:sstrails} yields the desired statement. So suppose
$u\neq v$. We may assume that $\lambda(u,v)\geq 2k$, since otherwise a minimum $(u,v)$-cut 
in $G$ is an odd $(u,v)$-trail cover of the required size. 
Let $E_{uv}$ be the $uv$ edge(s) in $G$ (which could be $\es$).
Let $\hG$ be obtained from $G-E_{uv}$ by identifying $u$ and $v$ into a new vertex $s$. 
(Note that $\hG$ has no loops.) 
Any odd $(u,v)$-trail in $G-E_{uv}$ 
maps to an odd $(s,s)$-trail in $\hG$. We apply Theorem~\ref{thm:sstrails} to 
$\hG, s, k'=k-|E_{uv}|$. 
If this returns an odd-$(s,s)$-trail cover $C$ of size at most $2k'-2$, then 
$C\cup E_{uv}$ is an odd-$(u,v)$-trail cover for $G$ of size at most $2k-2$. 
If we obtain a collection of $k'$ edge-disjoint odd $(s,s)$-trails in $\hG$, then
these together with $E_{uv}$ yield $k$ edge-disjoint odd $(\uv,\uv)$-trails in $G$.
Theorem~\ref{thm:contacts} then yields the required $k$ edge-disjoint odd $(u,v)$-trails. 
Polytime computability follows from the polytime computability in
Theorems~\ref{thm:sstrails} and~\ref{thm:contacts}. 
\end{proofof}

Theorem~\ref{thm:contacts} is our chief technical insight, which facilitates the
decoupling of the parity and $u$-$v$ connectivity requirements of odd $(u,v)$-trails,
thereby driving the entire proof. 
(It can be seen as a refinement of Theorem 5.1 in~\cite{Ibrahimpur16}.)
While Theorem~\ref{thm:sstrails} 
returns $(\uv,\uv)$-trails with the right
{\em parity}, 
Theorem~\ref{thm:contacts} supplies the missing ingredient needed to convert these into 
$(u,v)$-trails (of the same parity). 
We give an overview of the proofs of Theorems~\ref{thm:sstrails}
and~\ref{thm:contacts} below before delving into the details in the subsequent sections.   
We remark that both Theorem~\ref{thm:sstrails} and Theorem~\ref{thm:contacts} are 
{\em tight} as well, as we show in Appendices~\ref{append-sstexmpl}
and~\ref{append-contacts} respectively.  

\smallskip
The proof of Theorem~\ref{thm:contacts} relies on elementary arguments and proceeds as
follows (see Section~\ref{sec:contacts}).  
Let $\Pc$ be a collection of $2 \cdot |\hT|$ edge-disjoint $(u,v)$-paths. 
We provide a simple, efficient procedure to iteratively modify $\hT$
(whilst maintaining $|\hT|$ edge-disjoint odd $(\uv,\uv)$-trails) and eventually
obtain $|\hT|$ odd $(u,v)$-trails. 
Let $\Pc_0\sse\Pc$ be the collection of paths of $\Pc$ that are
edge-disjoint from trails in $\hT$.
First, we identify the trivial case where $|\Pc_0|$ is sufficiently large. If so, these
paths and $\hT$ directly yield odd $(u,v)$-trails as follows:
odd-length paths in $\Pc_0$ are already odd $(u,v)$-trails, and even-length paths in
$\Pc_0$ can be combined with odd $(u,u)$- and odd $(v,v)$- trails 
to obtain odd $(u,v)$-trails. 

The paths in $\Pc\sm\Pc_0$, all share at least one edge with some trail in $\hT$.
Each path is a sequence of edges from $u$ to $v$.
If the first edge that a path $P\in\Pc$ shares with a trail in $\hT$ lies on a
$(v,v)$-trail $T$, then it is easy to use parts of $P$ and $T$ to obtain an odd
$(u,v)$-trail that is edge-disjoint from all other trails in $\hT$,
and thereby make progress by increasing the number of odd $(u,v)$-trails in the
collection. A similar conclusion holds if the last edge that a path shares with a 
trail in $\hT$ lies on a $(u,u)$-trail. If neither of the above cases apply, then the
paths in $\Pc\sm\Pc_0$ are in a sense \emph{highly tangled} (which we formalize 
later) with trails in $\hT$. We then infer that $\Pc\sm\Pc_0$ and $\hT$ must satisfy some
simple structural properties, and leverage this to 
carefully modify the collection $\hT$ (while preserving edge-disjointness) so that the new
set of trails are ``less tangled'' with $\Pc$ than $\hT$, and thereby make progress.
Continuing this procedure a polynomial number of times yields the desired collection of
$|\hT|$ edge-disjoint odd $(u,v)$-trails. 

\smallskip
The proof of Theorem~\ref{thm:sstrails} relies 
on the key observation that 
we can cast our problem 
as the problem of packing and covering nonzero $A$-paths in a group-labeled graph 
$(H,\Gamma)$~\cite{ChudnovskyGGGLS06} for a suitable choice of $A, H$, and $\Gamma$ (see
Section~\ref{sec:sstrails}). In the latter problem,
(1) $H$ denotes an oriented graph whose arcs are labeled with elements of a group $\Gm$; 
(2) an $A$-path is a path in the undirected version of $H$ whose ends lie in 
$A$; and 
(3) the ($\Gm$-) length of an $A$-path $P$ is the sum of $\pm\gm_e$s (suitably defined)
for arcs in $P$, and a nonzero $A$ path is one whose length is not zero (where zero is the
identity element of $\Gm$).
Chudnovsky et al.~\cite{ChudnovskyGGGLS06} 
show that either there are $k$ {\em vertex-disjoint} non-zero $A$-paths, or there is a 
vertex-set of size at most $2k-2$ intersecting every non-zero $A$-path (Theorem 1.1
in~\cite{ChudnovskyGGGLS06}). 
We show that applying their result to a suitable ``gadget graph'' $H$ (essentially the 
line graph of $G$), yields Theorem~\ref{thm:sstrails} (see Section~\ref{sec:sstrails}).  
Polytime computability follows because a subsequent paper~\cite{ChudnovskyCG08} gave a
polytime algorithm for finding a maximum-size collection of vertex-disjoint non-zero
$A$-paths, and it is implicit in their proof that this also yields a suitable
vertex-covering of non-zero $A$-paths~\cite{Geelen16}. 

We remark that while the use of the packing-covering result in~\cite{ChudnovskyGGGLS06}
yields quite a compact proof of Theorem~\ref{thm:sstrails}, 
it also makes the resulting proof somewhat opaque 
since we apply the result in~\cite{ChudnovskyGGGLS06} to the gadget graph. 
However, it is possible to translate the min-max theorem for packing vertex-disjoint  
nonzero $A$-paths proved in~\cite{ChudnovskyGGGLS06} to our setting and obtain
the following more-accessible min-max theorem for packing edge-disjoint odd $(s,s)$-trails  
(stated in terms of $G$ and not the gadget graph).   
In Section~\ref{minmax}, we prove that

\begin{equation}
\nu(s,s;G) = \min \biggl(|E(S)-F| + 
\sum_{C\in\comp(G-S)} \floor{\frac{|E(S,C)|}{2}} \biggr) \label{mmeqn}
\end{equation}
where the minimum is taken over all bipartite subgraphs $(S,F)$ of $G$ such that 
$s\in S$. 

Notice that Theorem~\ref{thm:sstrails} follows easily from this min-max formula: if 
$(S^*,F^*)$ is a bipartite subgraph of $G$ with $s\in S^*$ that attains the minimum above,
then the edges in $E(S^*)-F$ combined with $\max\{0,|E(S^*,C)|-1\}$ edges from $E(S^*,C)$
for every component $C$ of $G-S^*$ yields a cover of size at most twice the RHS of
\eqref{mmeqn}.

\section{Proof of Theorem~\ref{thm:contacts}: converting edge-disjoint odd
$(\uv,\uv)$-trails to edge-disjoint odd $(u,v)$-trails} \label{sec:contacts}  
Recall that $\hT$ is a collection of edge-disjoint odd $(\uv,\uv)$-trails in $G$. We
denote the subset of odd $(u,u)$-trails, odd $(v,v)$-trails, and 
odd $(u,v)$-trails in $\hT$ by $\hT_{uu}$, $\hT_{vv}$, and $\hT_{uv}$, respectively. Let
$k_{uu}(\hT)=|\hT_{uu}|$, $k_{vv}(\hT)=|\hT_{vv}|$, and $k_{uv}(\hT)=|\hT_{uv}|$. 
To keep notation simple, we will drop the argument $\hT$ when its clear from the
context. Since we are given that $\lambda(u,v) \geq 2 \cdot |\hT|$, we 
can obtain a collection $\Pc$ of $2 \cdot |\hT|$ edge-disjoint $(u,v)$-paths in $G$.  
In the sequel, while we will modify our collection of odd $(\uv,\uv)$-trails, $\Pc$ stays
fixed.   

We now introduce the key notion of a {\em contact} between a trail $T$ and a $(u,v)$-path
$P$. Suppose that $P = (x_0,e_1,x_1, \dots, e_r, x_r)$ for some $r \geq 1$. 

\begin{definition} 
A \emph{contact} between $P$ and $T$ is a \emph{maximal subpath} $S$ of $P$ containing at
least one edge such that $S$ is also a subtrail of $T$ i.e., for 
$0\leq i < j \leq r$, we say that $(x_i,e_{i+1},x_{i+1}, \dots, e_j, x_j)$ is a contact
between $P$ and $T$ if $(x_i,e_{i+1},x_{i+1}, \dots, e_j, x_j)$ is a subtrail of $T$, but
neither $(x_{i-1},e_i,x_i, \dots e_j, x_j)$ (if $i > 0$) nor $(x_i,e_{i+1},x_{i+1},
\dots, x_j, e_{j+1}, x_{j+1})$ (if $j < r$) is a subtrail of $T$. 
\begin{equation*}
\text{Define} \quad
  \C(P,T) = \Bigl|\bigl\{(i,j): 0 \leq i < j \leq r, \quad 
(x_i,e_{i+1},x_{i+1}, \dots, e_j, x_j) 
\text{is a contact between $P$ and $T$}\bigr\}\Bigr|
\end{equation*}
\end{definition}

By definition, contacts between $P$ and $T$ are edge disjoint. 
For an edge-disjoint collection $\T$ of trails, 
we use $\C(P,\T)$ to denote $\sum_{T \in \T}\C(P,T)$. 
So if $\C(P,\T) = 0$, then $P$ is edge-disjoint from every trail in $\T$.  
Otherwise, 
we use the term \emph{first contact} of $P$ to refer to the contact arising
from the first edge that $P$ shares with some trail in $\T$ (note that $P$ is a
$(u,v)$-walk so is a sequence from $u$ to $v$). 
Similarly, the \emph{last contact} of $P$ is the contact arising from the last edge
that $P$ shares with some trail in $\T$. 
If $\C(P,\T)=1$, then the first and last contacts of $P$ are the same. 
We further overload notation and use $\C(\Pc,\T)$ to denote 
$\sum_{P \in \Pc}\C(P,\T) =\sum_{P \in \Pc,T \in \T} \C(P,T)$. 
We use $\C(\Pc,\T)$ as a measure of how ``tangled'' $\T$ is with $\Pc$. 
The following lemma 
classifies five different cases that arise for any pair of
edge-disjoint collections of 
{odd $(\uv,\uv)$-trails and $(u,v)$-paths.}

\begin{lemma} \label{lem:configs}
Let $\T$ be a collection of edge-disjoint odd $(\uv,\uv)$-trails in $G$. If 
$|\Pc| \geq 2 \cdot |\T|$, then one of the following conditions holds. 
\begin{enumerate}[(a), topsep=0.5ex, itemsep=0.5ex]
\item There are at least $k_{uu}(\T) + k_{vv}(\T)$ paths in $\Pc$ that make no contact with any trail in $\T$.
\item There exists a path $P \in \Pc$ that makes its first contact with a trail $T \in \Tv$.
\item There exists a path $P \in \Pc$ that makes its last contact with a trail $T \in \Tu$.
\item There exist three distinct paths $P_1,P_2,P_3 \in \Pc$ that make their first contact
  with a trail $T \in \Tu \cup \Tuv$.  
\item There exist three distinct paths $P_1,P_2,P_3 \in \Pc$ that make their last contact
  with a trail $T \in \Tuv \cup \Tv$. 
\end{enumerate}
\end{lemma}

\begin{proof}
To keep notation simple, we drop the argument $\T$ in the proof.
Suppose that conclusion (a) does not hold. 
Then there are at 
at least $2 \cdot |\T| - (k_{uu} + k_{vv} - 1) = 2k_{uv}+ k_{uu} + k_{vv} + 1$ paths in
$\Pc$ that make at least one contact with some trail in $\T$. Let $\Pc'\sse\Pc$ be this
collection of paths. 
If either conclusions (b) or (c) hold (for some $P\in\Pc'$), then we are done, so assume
that this is not the case.
Then, every path $P \in \Pc'$ makes its first contact with a trail in $\Tu\cup\Tuv$ 
and its last contact with a trail in $\Tuv \cup \Tv$. 
Note that the number of first and last contacts are both at least 
$2k_{uv} + k_{uu} + k_{vv} + 1>2\cdot\min(k_{uv}+k_{uu},k_{uv}+k_{vv})$.  
So if $k_{uu}\leq k_{vv}$, then by the Pigeonhole principle, there are at least 3 paths
that make their first contact with some $T\in\T_{uu}\cup\T_{uv}$, i.e., conclusion (d)
holds. Similarly, if $k_{vv}\leq k_{uu}$, then conclusion (e) holds.
\end{proof}

We now leverage the above classification and show that in each of the above five cases, 
we can make progress by ``untangling'' the trails (i.e., decreasing $C(\Pc,\T)$) and/or
increasing the number of odd $(u,v)$-trails in our collection. 

\begin{lemma} \label{lem:potential}
Let $\T$ be a collection of edge-disjoint odd $(\uv,\uv)$-trails. If 
$|\Pc|\geq 2 \cdot |\T|$, we can obtain another collection $\T'$ of edge-disjoint odd
$(\uv,\uv)$-trails such that at least one of the following holds.
\begin{enumerate}[(i), topsep=0.5ex, itemsep=0.5ex, leftmargin=6ex]
\item $k_{uv}(\T') = |\T|$.
\item $\C(\Pc,\T') \leq \C(\Pc,\T)$ and $k_{uv}(\T') = k_{uv}(\T) + 1$.
\item $\C(\Pc,\T') \leq \C(\Pc,\T) - 1$ and $k_{uv}(\T') \geq k_{uv}(\T) - 1$.
\end{enumerate}
\end{lemma}

\begin{proof}
If $k_{uv}(\T) = |\T|$, then (i) holds trivially by taking $\T' = \T$. 
So we may assume that $\T$ contains some odd $(u,u)$- or odd $(v,v)$-trail. 
Observe that $\T$ and $\Pc$ satisfy the conditions of Lemma~\ref{lem:configs}, so at least
one of the five conclusions of Lemma~\ref{lem:configs} applies. We handle each case
separately.   

\begin{enumerate}[(a)]
\item At least $k_{uu}(\T) + k_{vv}(\T)$ paths in $\Pc$ have zero contacts with $\T$. 
Let $\Pc_0=\{P\in\Pc: \C(P,\T)=0\}$. 
Consider some $P \in \Pc_0$. If $P$ is odd, 
we can replace an odd $(u,u)$- or odd $(v,v)$- trail in $\T$ with $P$. If $P$ is
even, then $P$ can be combined with an odd $(u,u)$- or odd $(v,v)$- trail to obtain an odd
$(u,v)$-trail. Since $|\Pc_0| \geq k_{uu}(\T) + k_{vv}(\T)$, we create
$k_{uu}(\T)+k_{uv}(\T)$ odd $(u,v)$-trails this way, and this new collection $\T'$ 
satisfies (i).

\item Some $P \in \Pc$ makes its first contact with an odd $(v,v)$-trail $T \in \T$. Let
the first vertex in the first contact between $P$ and $T$ be $x$. Observe that $x$
partitions the trail $T$ into two subtrails $S_1$ and $S_2$. Since $T$ is an odd trail,
exactly one of $S_1$ and $S_2$ is odd. We can now obtain an odd $(u,v)$-trail $T'$ 
by traversing $P$ from $u$ to $x$, 
and then traversing $S_1$ or $S_2$, whichever yields odd parity (see Fig.~\ref{casebfig}). 
Since $P$ already made a contact with $T$, we have $\C(P,T') \leq \C(P,T)$, 
and $\C(Q,T')\leq \C(Q,T)$ for any other path $Q \in \Pc$. 
Thus, taking $\T' = (\T \cup \{ T' \}) \setminus \{ T \}$, we have 
$\C(\Pc,\T')\leq \C(\Pc,\T)$, and (ii) holds.

\begin{figure}[h!]
\centering
\begin{tikzpicture}

\setcounter{xx}{0}
\setcounter{yy}{0}
\coordinate (centre) at (\value{xx},\value{yy});
\node (dummy) at (centre) {};

\coordinate (u) at ([xshift=-4cm,yshift=-3cm]centre);
\coordinate (v) at ([xshift=4cm]u);
\coordinate (xa) at ([xshift=-1cm,yshift=1.7cm]v);
\coordinate (xb) at ([xshift=0.25cm,yshift=1cm]xa);

\draw [fill=black] (u) circle (\rad) node[left] {$u$};
\draw [fill=black] (v) circle (\rad) node[below] {$v$};
\draw [fill=black] (xa) circle (\rad) node[right] {$x$};
\draw [fill=black] (xb) circle (\rad);

\begin{scope}[yscale=-1,xscale=1]
\draw[looseness=0.75] (-1,1) arc (180:360:1) to[out=90,in=0] (0,3) to[out=180,in=90] (-1,1) -- cycle;  
\end{scope}

\draw (u) edge[bend right=20,line width=1mm,gray] (xa);
\draw (xa) edge[bend left=10,line width=1mm,gray] (xb);
\draw (xb) edge[bend left,dashed] (v);

\node (traila) at ([xshift=4.5cm]xa) {$T_1 = u \xrightarrow[P]{} x \xrightarrow[S_1]{} v$};
\node (trailb) at ([yshift=-1cm]traila) {$T_2 = u \xrightarrow[P]{} x \xrightarrow[S_2]{} v$};
\node (lega) at ([xshift=-0.1cm,yshift=-1cm]xa) {$S_1$};
\node (legb) at ([xshift=2.3cm,yshift=0.5cm]xa) {$S_2$};
\node (puv) at ([xshift=1.5cm,yshift=0.8cm]u) {$P$};

\end{tikzpicture}
\caption{Path $P$ makes its first contact with an odd $(v,v)$-trail.} 
\label{fig:vvtrail} \label{casebfig} 
\end{figure}
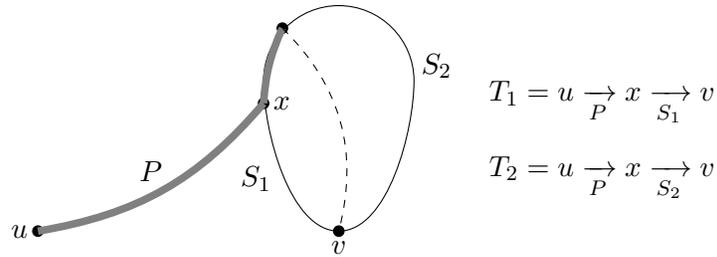

\item Some $P \in \Pc$ makes its last contact with an odd $(u,u)$-trail $T \in \T$. This
is completely symmetric to (b), so a similar strategy works and we satisfy (ii).

\item Paths $P_1,P_2,P_3 \in \Pc$ that make their first contact with an odd
$(u,\uv)$-trail $T\in\T$. Note that all contacts between paths in $\Pc$ and trails in $\T$ 
are edge disjoint, since the paths in $\Pc$ are edge disjoint and the trails in $\T$ are
edge disjoint. For $i=1,2,3$, let the first vertex in the first
contact of $P_i$ (with $T$) be $x_i$. Let $Q_i$ denote the subpath of $P_i$ between $u$
and $x_i$. Note that $T$ is a sequence of edges from $u$ to some vertex in $\uv$.
Without loss of generality, assume that in $T$, 
the first contact of $P_1$ appears before the first contact of $P_2$, which appears before
the first contact of $P_3$. 
The vertices $x_1,x_2,x_3$ partition the trail $T$ into four subtrails $S_0,S_1,S_2,S_3$
(see Fig.~\ref{fig:3contacts}). 
For a trail $X$, we denote the reverse sequence of $X$ by $\overline{X}$. Now consider
the following trails (where $+$ denotes concatenation): 
$$
T_1=S_0+\bQ_1, \quad T_2=Q_1+S_1+\bQ_2, \quad T_3=Q_2+S_2+\bQ_3, \quad
T_4=Q_3+S_3.
$$
Observe that the disjoint union of edges in $T_1,T_2,T_3,$ and $T_4$ has the same parity
as that of $T$, and hence at least one of the $T_i$s 
is an odd trail; call this trail $T'$. 
Let $\T'=\T\cup\{T'\}\sm\{T\}$.
By construction, {\em every} $T_i$ avoids at least one of the (first) contacts made by
$P_1$, $P_2$, or $P_3$ (with $T$). Also, for any other path 
$Q \in \Pc \setminus\{P_1,P_2,P_3\}$, we have $\C(Q,T') \leq \C(Q,T)$. 
Therefore, $\C(\Pc,\T') \leq \C(\Pc,\T)-1$. It could be that $T$ was an odd $(u,v)$-trail,
which is now replaced by an odd $(u,u)$-trail, so $k_{uv}(\T')\geq k_{uv}(\T)-1$.
{So we satisfy (iii).}

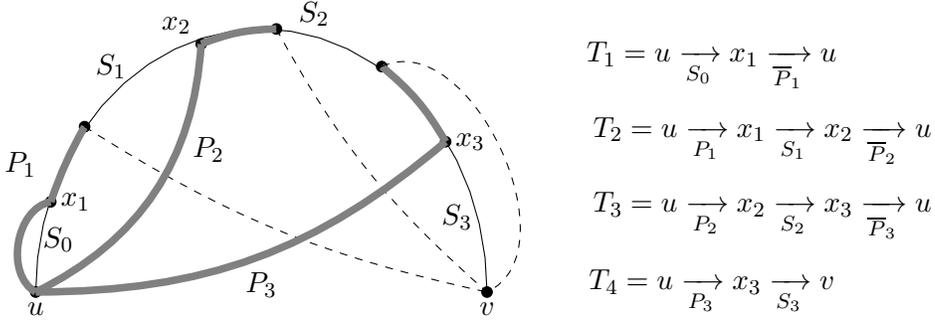
\begin{figure}[ht!]
\centering
\begin{tikzpicture}

\setcounter{xx}{0}
\setcounter{yy}{0}
\coordinate (centre) at (\value{xx},\value{yy});
\node (dummy) at (centre) {};

\coordinate (u) at ([xshift=-3cm,yshift=0cm]centre);
\coordinate (v) at ([xshift=6cm]u);
\coordinate (xa) at ([xshift=0.2cm,yshift=1.2cm]u);
\coordinate (xb) at ([xshift=0.45cm,yshift=1cm]xa);
\coordinate (ya) at ([xshift=2.2cm,yshift=3.3cm]u);
\coordinate (yb) at ([xshift=1cm,yshift=0.2cm]ya);
\coordinate (za) at ([xshift=5.45cm,yshift=2cm]u);
\coordinate (zb) at ([xshift=-0.85cm,yshift=1cm]za);

\draw [fill=black] (u) circle (\rad) node[below] {$u$};
\draw [fill=black] (v) circle (\rad) node[below] {$v$};
\draw [fill=black] (xa) circle (\rad) node[right] {$x_1$};
\draw [fill=black] (xb) circle (\rad);
\draw [fill=black] (ya) circle (\rad) node[above left] {$x_2$};
\draw [fill=black] (yb) circle (\rad);
\draw [fill=black] (za) circle (\rad) node[right] {$x_3$};
\draw [fill=black] (zb) circle (\rad);

\draw [domain=0:180] plot ({3*cos(\x)}, {3.5*sin(\x)});

\draw (u) edge[bend left=70,line width=1mm,gray] (xa);
\draw (xa) edge[bend left=5,line width=1mm,gray] (xb);
\draw (xb) edge[bend right=10,dashed] (v);

\draw (u) edge[bend right=30,line width=1mm,gray] (ya);
\draw (ya) edge[bend left=10,line width=1mm,gray] (yb);
\draw (yb) edge[bend right=10,dashed] (v);

\draw (u) edge[bend right=20,line width=1mm,gray] (za);
\draw (za) edge[bend right=10,line width=1mm,gray] (zb);
\draw (zb) edge[bend left=90,dashed] (v);

\node (pa) at ([xshift=-0.2cm,yshift=1.7cm]u) {$P_1$};
\node (pb) at ([xshift=0.1cm,yshift=-1.4cm]ya) {$P_2$};
\node (pc) at ([xshift=0cm,yshift=0.1cm]centre) {$P_3$};
\node (sa) at ([xshift=0.3cm,yshift=0.7cm]u) {$S_0$};
\node (sb) at ([xshift=-1.2cm,yshift=-0.3cm]ya) {$S_1$};
\node (sc) at ([xshift=1.5cm,yshift=0.4cm]ya) {$S_2$};
\node (sd) at ([xshift=-0.4cm,yshift=1cm]v) {$S_3$};

\node (ta) at ([xshift=4cm,yshift=3cm]v) {$T_1 = u \xrightarrow[S_0]{} x_1 \xrightarrow[\overline{P}_1]{} u$ \gapten};
\node (tb) at ([yshift=-1cm]ta) {$T_2 = u \xrightarrow[P_1]{} x_1 \xrightarrow[S_1]{} x_2 \xrightarrow[\overline{P}_2]{} u$ \gaptwo \gap};
\node (tc) at ([yshift=-1cm]tb) {$T_3 = u \xrightarrow[P_2]{} x_2 \xrightarrow[S_2]{} x_3 \xrightarrow[\overline{P}_3]{} u$ \gaptwo \gap};
\node (td) at ([yshift=-1cm]tc) {$T_4 = u \xrightarrow[P_3]{} x_3 \xrightarrow[S_3]{} v$ \gapten};

\end{tikzpicture}
\caption{Paths $P_1,P_2,P_3$ make their first contact with an odd $(u,v)$-trail.} 
\label{fig:3contacts} \label{casedfig}
\vspace*{-2ex}
\end{figure}

\item Paths $P_1,P_2,P_3 \in \Pc$ make their last contact with an odd $(\{u,v\},v)$-trail
in $\T$. This is completely symmetric to (d) and the same approach works, 
{so we again satisfy (iii).} 
\end{enumerate}
\vspace*{-2ex}
\end{proof}

Theorem~\ref{thm:contacts} now follows by simply applying Lemma~\ref{lem:potential}
starting with the initial collection $\T^0:=\hT$ until conclusion (i) of
Lemma~\ref{lem:potential} applies. 
The $\T'$ returned by this final application of Lemma~\ref{lem:potential} then
{satisfies the theorem statement.}

We now argue that this process terminates in at most $2 \cdot |E(G)|+|\hT|$ steps, which
will conclude the proof. Let $k=|\hT|$. Consider the following potential function defined
on a collection $\T$ of $k$ edge-disjoint odd $(\uv,\uv)$-trails: 
$\phi(\T):=2\cdot \C(\Pc,\T) - k_{uv}(\T)$. 
Consider any iteration where we invoke Lemma~\ref{lem:potential} and move from a
collection $\T$ to another collection $\T'$ with $k_{uv}(\T') < k$. Then, either
conclusion (ii) or (iii) of Lemma~\ref{lem:potential} applies, 
and it is easy to see that $\Phi(\T')\leq \phi(\T)-1$. 
Finally,  we have $-k \leq \Phi(\T) \leq 2 \cdot |E(G)|$ for all $\T$ since 
$0 \leq \C(\Pc,\T) \leq |E(G)|$ 
as the contacts between paths in $\Pc$ and trails in $\T$ are edge-disjoint, so the
process terminates in at most $2|E(G)|+k$ steps.

\section{Proof of Theorem~\ref{thm:sstrails}} \label{sec:sstrails}
Our proof relies on two reductions both involving non-zero $A$-paths in a group-labeled
graph, which we now formally define.
A {\em group-labeled graph} is a pair $(H,\Gm)$, where $\Gm$ is a group, and $H=(N,E')$ is
an oriented graph (i.e., for any $u,v\in N$, if $(u,v)\in E'$ then $(v,u)\notin E'$)
whose arcs are labeled with elements of $\Gm$. All addition (and 
subtraction) operations below are always with respect to the group $\Gm$. 
A path $P$ in $H$ is a sequence $(x_0,e_1,x_1,\ldots,e_r,x_r)$, where the $x_i$s are distinct, 
and each $e_i$ has ends $x_i$, $x_{i+1}$ but could be oriented either way (i.e., as
$(x_i,x_{i+1})$ or $(x_{i+1},x_i)$). 
(So upon removing arc directions, $P$ yields a path in the undirected version of $H$.) 
We say that $P$ traverses $e_i$ in the direction $(x_i,x_{i+1})$.
The {\em $\Gm$-length} (or simply length) of $P$, denoted $\gm(P)$, is the sum of
$\pm\gm_e$s for arcs in $P$, where we count $+\gm_e$ for $e$ if $P$'s traversal of $e$
matches $e$'s orientation and $-\gm_e$ otherwise.
Given $A\sse N$, an $A$-path is a path $(x_0,e_1,\ldots,e_r,x_r)$ where $r\geq 1$, and
$x_0,x_r\in A$; finally, call an $A$-path $P$ a nonzero $A$-path if $\gm(P)\neq 0$ (where
$0$ denotes the identity element for $\Gm$).

Chudnovsky et al.~\cite{ChudnovskyGGGLS06} proved the following theorem as a consequence
of a min-max formula they obtain for the maximum number of vertex-disjoint nonzero
$A$-paths. Subsequently, \cite{ChudnovskyCG08} devised a polytime algorithm to compute   
the maximum number of vertex-disjoint non-zero $A$-paths. Their algorithm also implicitly
computes the quantities needed in (the minimization portion of) their min-max formula to
show the optimality of the collection of $A$-paths they return~\cite{Geelen16}; this in
turn easily yields the vertex-set mentioned in Theorem~\ref{grpthm}.

\begin{theorem}[\cite{ChudnovskyGGGLS06,ChudnovskyCG08}] \label{grpthm}
Let $\bigl(H=(N,E'),\Gm\bigr)$ be a group-labeled graph, and $A\sse V$. Then, for any
integer $k$, one can obtain in polynomial time, either:
\begin{enumerate}[topsep=0.5ex, itemsep=0.5ex]
\item $k$ vertex-disjoint nonzero $A$-paths, or
\item a set of at most $2k-2$ vertices that intersects every nonzero $A$-path.
\end{enumerate}
\end{theorem}

Recall that $G$ is the undirected graph in the theorem statement, and $s\in V$.
For a suitable choice of a group-labeled graph $(H,\Gm)$, and a vertex-set $A$, we show
that: (a) 
vertex-disjoint nonzero $A$-paths in $(H,\Gm)$ yield edge-disjoint odd $(s,s)$-trails;   
and
(b) a {\em vertex-set} covering all nonzero $A$-paths in $(H,\Gm)$ yields an odd
$(s,s)$-trail cover of $G$.
Combining this with Theorem~\ref{grpthm} finishes the proof.

Since we are dealing with parity, it is natural to choose $\Gamma=\grp$ (so the
orientation of edges in $H$ will not matter).
To translate vertex-disjointness (and vertex-cover) to edge-disjointness (and edge-cover),
we essentially work with the line graph of $G$, but slightly modify it to incorporate
edge labels. 
We replace each vertex 
$x\in V$ with a clique of size $\deg_G(x)$, with each clique-node corresponding to a
distinct edge of $G$ incident to $x$; we use $[x]$ to denote this clique, both its
set of nodes and edges; the meaning will be clear from the context. (Note that if
$\deg_G(x)=0$, then there are no nodes and edges corresponding to $x$ in $H$; this is fine
since isolated nodes in $G$ can be deleted without affecting anything.) 
For every edge $e=xy\in E$, we create an edge between the clique nodes of $[x]$ and $[y]$
corresponding to $e$. We arbitrarily orient the edges to obtain $H$. We give each clique
edge a label of $0$, and give every other edge a label of $1$. 
Finally, we let $A=[s]$. 

\begin{lemma} \label{grpcorres}
The following properties hold.
\begin{enumerate}[(a), topsep=0.5ex, itemsep=0.5ex]
\item Every $A$-path $P$ in $H$ maps to an $(s,s)$-trail $T=\pi(P)$ in $G$ such that $\gm(P)=1$
iff $T$ is an odd trail.
\item If two $A$-paths $P, Q$ are vertex disjoint then the $(s,s)$-trails $\pi(P)$ and
$\pi(Q)$ are edge disjoint.  
\item Every $(s,s)$-trail $T$ in $G$ with at least one edge maps to an $A$-path $P=\sg(T)$
in $G$ such that: $T$ is an odd trail iff $\gm(P)=1$, and $P$ contains a vertex $x$ iff
$T$ contains the corresponding edge of $G$.
\end{enumerate}
\end{lemma}

\begin{proof} 
The proof is straightforward.
By definition, the 1-labeled edges of $H$ are in bijective correspondence with the edges
of $G$. 

For part (a), let $P=(x_0,e_1,\ldots,e_r,x_r)$ be an $A$-path in $H$. 
Let $P'=(e_{i_1},\ldots,e_{i_q})$ be the subsequence of $P$ consisting of the
1-labeled edges of $P$. If $P'=()$, then define $\pi(P)$ to be the trivial $(s,s)$-trail
$(s)$, which satisfies the required property. 
Let $f_j$ be the edge of $G$ corresponding to $e_{i_j}$.
Then, for $j=2,\ldots,q-1$, we have $f_j=u_{j-1}u_j$, where $u_{j-1}$ is such that
$[u_{j-1}]$ contains one end of both $e_{i_j}$ and $e_{i_{j-1}}$, and $u_j$ is such that
$[u_j]$ contains one end of both $e_{i_j}$ and $e_{i_{j+1}}$.
The end of $e_{i_1}$ not in $[u_1]$ must lie
in $[s]$ (as all edges of $P$ occurring before $e_{i_1}$ have label 0, and so must be
edges of $[s]$). Similarly, the end of $e_{i_q}$ that does not lie in $[u_{q-1}]$ must
lie in $[s]$ (as all edges of $P$ occurring after $e_{i_q}$ must be edges of $[s]$).
Therefore, the sequence $f_1,\ldots,f_q$ yields an $(s,s)$-trail $\pi(P)$ in $G$, and we 
have $|\pi(P)|=|P'|=\gm(P)$.
Observe that an edge $f$ of $G$ lies in $\pi(P)$ iff the corresponding 1-labeled edge of
$H$ lies in $P$.

For part (b), 
let $T=\pi(P)$, $T'=\pi(Q)$ be the corresponding $(s,s)$-trails. We may assume that $T$
and $T'$ contain at least one edge, otherwise the statement is vacuously true. 
Since $P$ and $Q$ are vertex disjoint, their subsequences of 1-labeled edges, which map to
the edges of $T$ and $T'$, do not contain any common edges, so $T$ and $T'$ are edge
disjoint. 

For part (c), consider an $(s,s)$-trail $(u_0=s,f_1,\ldots,f_r,u_r=s)$, where $r\geq 1$. Each
$f_i$ maps to a distinct 1-labeled edge $e_i$ of $H$. The edges $e_i$ and $e_{i+1}$ both
have one end incident to a distinct node in $[u_i]$, and there is a clique edge joining
these ends. Since the $f_i$s are all distinct and no two 1-labeled edges of $H$ share an
endpoint, the ends of all the $e_i$s are distinct. 
Let $x_0$ be the end of $e_1$ in $[u_0]=[s]$, and $x_r$ be the end of $e_r$ in
$[u_r]=[s]$. Note that $x_0\neq x_r$ since $e_1\neq e_r$. 
So the edges $e_i$ interspersed with the clique edges from $[u_i]$ yield an
$(x_0,x_r)$-path $\sg(T)$ in $H$, which is therefore an $A$-path.
Also, $\gm(P)=r\bmod 2$, so $P$ is a nonzero $A$-path iff $T$ is an odd trail.
This proves part (ii). By construction, we ensure that a node $x$ lies in $P$ iff the 
corresponding edge $f$ of $G$ lies in $T$. 
\end{proof}

To complete the proof of Theorem~\ref{thm:sstrails}, we apply Theorem~\ref{grpthm} to 
the nonzero $A$-path instance $(H,[s],\gm,\grp)$ constructed above. If we obtain $k$
vertex-disjoint nonzero $A$-paths in $H$, then parts (a) and (b) of Lemma~\ref{grpcorres}
imply that we can map these to $k$ edge-disjoint odd $(s,s)$-trails.
Alternatively, if we obtain a set $C$ of at most $2k-2$ vertices of $H$ that intersect
every nonzero $A$-path, then we obtain a cover $F$ for odd $(s,s)$-trails in $G$ 
by taking the set of edges in $G$ corresponding to the vertices in $C$. 
To see why $F$ is a cover, suppose that the graph $G-F$ has an odd $(s,s)$-trail. This
then maps to a nonzero $A$-path $P$ in $H$ such that $P\cap C=\es$ by part (c) of
Lemma~\ref{grpcorres}, which yields a contradiction.

\section{Min-max theorem for packing odd edge-disjoint \boldmath $(s,s)$-trails} \label{minmax} 
Let $G$ be an undirected graph. For a node $s\in V(G)$, recall that $\nu(s,s;G)$
denotes the maximum number of edge-disjoint odd $(s,s)$-trails in $G$. Let $\comp(G)$
denote the set of components of $G$. Given disjoint vertex sets $S, T$, we use $E_G(S,T)$
to denote the set of edges in $G$ with one end in $S$ and one end in $T$. 
Let $E_G(S)$ be the set of edges with both ends in $S$. We drop the subscript $G$ above
when it is clear from the context. 
We prove the following min-max theorem. 

\begin{theorem} \label{minmaxthm}
Let $G$ be an undirected graph, and $s$ be a node in $V(G)$. Then, 
\begin{equation}
\nu(s,s;G) = \min \biggl\{
|E(S)-F|+\negthickspace\sum_{C\in\comp(G-S)}\floor{\frac{|E(S,C)|}{2}}: 
\ \ s\in S,\ (S,F)\text{ is a  bipartite subgraph of }G\biggr\}. 
\label{minmaxeqn} 
\end{equation}
\end{theorem}

In the sequel, we fix $G$ and $s$ to be the graph and the node mentioned in the statement
of Theorem~\ref{minmaxthm}. It is easy to see that for any bipartite subgraph $(S,F)$ with
$s\in S$, the expression in the RHS of \eqref{minmaxeqn} gives an upper
bound on $\nu(s,s;G)$: an odd $(s,s)$-trail either lies completely within $S$ and must
therefore use some edge of $E(S)\sm F$, or must include a vertex of some component
$K\in\comp(G-S)$; there are at most $E(S)\sm F$ edge-disjoint trails of the first type,
and at most $\floor{\frac{|E(S,K)|}{2}}$ edge-disjoint trails that may use a vertex of a
component $K$. 
Our goal is to show that for some choice of $(S,F)$, the expression in the RHS of
\eqref{minmaxeqn} equals $\nu(s,s;G)$. 

Let $(H,[s],\gm:E(H)\mapsto\{0,1\},\grp)$ be the nonzero $A$-path packing instance
obtained from $(G,s)$ as described in Section~\ref{sec:sstrails}. Let $\nu(H,[s],\gm)$
denote the maximum number of vertex-disjoint non-zero $[s]$-paths in $H$ under the
labeling $\gm$, so $\nu(s,s;G)=\nu(H,[s],\gm)$.    
As mentioned earlier, we obtain our min-max theorem by applying the min-max formula
of Chudnovsky et al.~\cite{ChudnovskyGGGLS06} for packing vertex-disjoint nonzero
$A$-paths in group-labeled graphs to $(H,[s],\grp)$, 
and simplifying the resluting expression by leveraging the structure underlying the gadget
graph $H$. 
Our starting point therefore is the following result obtained by
specializing Theorem 1.2 in~\cite{ChudnovskyGGGLS06} to the $A$-path instance
$(H,[s],\gm,\grp)$. We need the following notation from~\cite{ChudnovskyGGGLS06}.
Given $y\in V(H)$, {\em switching} the labeling $\gm$ at $y$ means that we flip the labels
of the edges incident to $y$; that is, we obtain a new labeling $\gm'$ where
$\gm'_e=1-\gm_e$ if $e$ is incident to $y$ and $\gm'_e=\gm_e$ otherwise. 
Note that if $y\notin[s]$, then $\nu(H,[s],\gm)=\nu(H,[s],\gm')$.
For $A'\sse V(H)$ and a labeling $\gm'$, let $E(A',\gm')$ denote the set of all edges
$e\in E(H)$ with both ends in $A'$ and having $\gm'_e=0$.

\begin{theorem}[Corollary of Theorem 1.2 in~\cite{ChudnovskyGGGLS06}] 
\label{thm:groupminmax} \label{grpminmaxthm}
There exists a labeling $\gm':E(H)\mapsto\{0,1\}$ obtained by switching $\gm$ at some
(suitably chosen) vertices in $V(H)-[s]$, and vertex-sets $X, A' \subseteq V(H)$ with
$[s]-X \subseteq A' \subseteq V(H)-X$ such that
\begin{equation}
\nu(H,[s],\gamma) = |X| + \sum_{K \in \comp(H-X-E(A',\gm'))} \floor{ \frac{|A' \cap V(K)|}{2} }.
\label{grpminmaxeqn}
\end{equation}
\end{theorem}

We may assume that the labeling $\gm'$ given above is obtained by switching $\gm$ at some
subset of $A'-[s]$, since we can always switch $\gm'$ at vertices not in $A'$ without
affecting the RHS of \eqref{grpminmaxeqn}. It will be convenient to restate
Theorem~\ref{grpminmaxthm} by explicitly referring to the subset of $A'$ at which $\gm$
was switched, as follows. 
Let $Y, B_0, B_1 \subseteq V(H)$ be disjoint sets such that 
$[s]-Y\sse B_0$ and $B_0 \cup B_1\subseteq V(H)-Y$; we call $(Y,B_0,B_1)$ a
valid triple. Let $\gamma(B_1)$ denote the labeling obtained from $\gamma$ by switching at
every vertex in $B_1$. Let $H(Y,B_0,B_1):=H-Y- E(B_0 \cup B_1,\gamma(B_1))$. Define  
$$
p(Y,B_0,B_1) := |Y| + \sum_{K\in\comp(H(Y,B_0,B_1))} \floor{ \frac{|(B_0 \cup B_1) \cap V(K)|}{2} }.
$$
As noted in~\cite{ChudnovskyGGGLS06}, it is not hard to see that 
$\nu(H,[s],\gm)\leq p(Y,B_0,B_1)$ for any valid triple $(Y,B_0,B_1)$. We call a valid 
triple $(Y,B_0,B_1)$ \emph{tight} if $\nu(H,[s],\gm)=p(Y,B_0,B_1)$. 
Let the labeling $\gm'$ given by Theorem~\ref{grpminmaxthm} be obtained by switching $\gm$ 
at $A'_1\sse A'$. Let $A'_0:=A'-A'_1$. Theorem~\ref{thm:groupminmax} then tells us that
$(X,A'_0,A'_1)$ is a tight triple. Exploiting the structure of $H$, we show that there
exists a tight triple of a convenient form, which will then lead to our min-max formula.

\begin{lemma} \label{lem:emptyx}
There exists a tight triple $(\es,B_0,B_1)$ such that, for every $x\in V(G)$, we have
\begin{equation} 
[x] \subseteq B_0 \quad \text{or} \quad [x] \subseteq B_1 \quad \text{or} \quad
[x]\cap(B_0\cup B_1)=\es. \label{eq:emptyx}
\end{equation}
\end{lemma}

Before proving Lemma~\ref{lem:emptyx}, we show how this yields Theorem~\ref{minmaxthm}.

\begin{proofof}{Theorem~\ref{minmaxthm}}
Let $(\es,B_0,B_1)$ be the tight triple given by Lemma~\ref{lem:emptyx}. Let
$\tgm=\gm(B_1)$.
Take $S=S_0\cup S_1$, where $S_0 := \{x \in V(G): [x] \subseteq B_0\}$ and 
$S_1 := \{x \in V(G): [x]\subseteq B_1\}$. Note that $s \in S_0$. Let 
$F :=E_G(S_0,S_1)$. Let $C_1,\dots,C_{\ell}$ be the components of $G-S$. We show that 
$p(\es,B_0,B_1)=|E(S)-F|+\sum_{i=1}^\ell\floor{\frac{|E(S,C_i)|}{2}}$, which together with
the fact that $(\es,B_0,B_1)$ is a tight triple, completes the proof.

Observe that for every vertex $x \in S$, all the clique edges of
$H$ in $[x]$ have $\tgm$-label 0 (since $\gm$ is switched at either both endpoints or no
endpoint) and thus belong to $E(B_0 \cup B_1,\gamma(B_1))$. Therefore such 
edges do not appear in $\tH:=H(\es,B_0,B_1)$. Also, for every edge $e = xy \in F$, there is
a corresponding edge $e$ in $H$ between some vertex $w$ in $B_0$ and some vertex $z$ in
$B_1$, and $\tgm_e=0$; so $e\notin E_{\tH}$. Next, for every 
$e = xy \in E(S)-F=E(S_0) \cup E(S_1)$, there is a corresponding edge in $\tH$
between some vertex $w\in [x]$ and some vertex $z\in[y]$, and this edge has $\tgm$-label
1. Thus, $\{w,z\}$ is a connected component of $\tH$ of size $2$, and contributes $1$ to
$p(\emptyset,B_0,B_1)$. 

We now argue that components $\{C_i\}_{i=1}^{\ell}$ are in bijection with the remaining
components of $\tH$, and if $C_i$ corresponds to component $K_i$ of $\tH$, then
$K_i$ contributes $\floor{\frac{|E(S,C_i)|}{2}}$ to $p(\es,B_0,B_1)$.
Consider any component $C_i$. Let $E(S,C_i)=\{x_1 y_1, \dots, x_r y_r\}$, 
where $x_j \in S$ and $y_j \in C_i$ for all $j=1,\ldots,r$. 
For $j=1,\dots,r$, let $w_j\in[x_j]$ be the end in $B_0\cup B_1$ of the unique edge in
$\tH$ corresponding to edge $x_jy_j$. 
For every $y_j$, all vertices of $H$ in $[y_j]$ belong to the same component of $\tH$.
So since $y_1,\ldots,y_r$ are connected in $G-S$, 
all vertices in $\bigcup_{j=1}^r [y_j]$ belong to the same component of
$\tH$, say $K_i$. It follows that
$K_i\supseteq\{w_1,\ldots,w_r\}\cup\bigl(\bigcup_{j=1}^r[y_j]\bigr)$. 
Thus we have, 
$$
\floor{\frac{|E(S,C_i)|}{2}}=\floor{\frac{r}{2}}=\floor{\frac{|V(K_i) \cap (B_0\cup B_1)|}{2}}.
$$

Moreover, since $G-S=\bigcup_{i=1}^\ell C_i$, it follows that $\tH$ has no components
other than $K_1,\ldots,K_\ell$. So to summarize, we have shown that 
$p(\es,B_0,B_1)=|E(S)-F|+\sum_{i=1}^\ell\floor{\frac{|E(S,C_i)|}{2}}$, completing the
proof. 
\end{proofof}

We now prove Lemma~\ref{lem:emptyx}. We will need the following claim.

\begin{claim} \label{basic}
For any integer $q \geq 2$ and nonnegative integers $r_1,\dots,r_{q}$ we have 
\begin{equation}\label{eq:floor1}
\sum_{i=1}^{q} \floor{ \frac{r_i}{2} } \geq \floor{ \frac{\sum_{i=1}^{q} r_i - (q-1)}{2} } \ .
\end{equation}
\end{claim}

\begin{proof}
For any integer $x$, we have $\floor{\frac{x}{2}}\geq\frac{x-1}{2}$, and equal to $\frac{x}{2}$
if $x$ is even. So if any of the $r_i$'s are even, the statement follows. Otherwise,
$\sum_{i=1}^q r_i-q$ is even, and both the LHS and RHS of \eqref{eq:floor1} are equal
to $\frac{\sum_{i=1}^q r_i-q}{2}$. 
\end{proof}

\begin{proofof}{Lemma~\ref{lem:emptyx}}
Let $(X,A'_0,A'_1)$ be the tight triple given by Theorem~\ref{grpminmaxthm}.
First, we show how to modify $(X,A'_0,A'_1)$ to obtain a tight triple
$(X,\tA_0,\tA_1)$ that satisfies  
\begin{equation} \label{eq:uniform}
([y]-X) \subseteq \tA_0 \quad \text{or} \quad
([y]-X) \subseteq \tA_1 \quad \text{or} \quad
([y]-X) \subseteq V(H(X,\tA_0,\tA_1))-(\tA_0 \cup \tA_1)
\end{equation}
for every $y \in V(G)$. Then, we show how vertices in $X$ can be appropriately moved out
of $X$ to obtain the required tight triple $(\emptyset,B_0,B_1)$. 

\paragraph{Step 1: satisfying \eqref{eq:uniform}.} 
Let $A'_*$ denote $V(H(X,A'_0,A'_1)) - A'_0 - A'_1=V(H) - X - A'_0 - A'_1$. Suppose there
is a vertex $y\in V(G)$ that violates \eqref{eq:uniform}, so $([y]-X)\not\sse T$ for
$T\in\{A'_0,A'_1,A'_*\}$. Note that $y\neq s$ since $[s]\sse A'_0$ (as $(X,A'_0,A'_1)$ is
a valid triple).
Define $D_0 := A'_0-[y]$, $D_1 := A'_1-[y]$ and 
$D_* := V(H)-X-D_0-D_1$. Clearly, $(X,D_0,D_1)$ is a valid triple, and
$[y]-X\sse D_*$.
We show that $(X,D_0,D_1)$ is a tight triple; so by repeatedly applying the above
modification for every vertex that violates \eqref{eq:uniform}, we obtain a tight triple
$(X,\tA_0,\tA_1)$ satisfying \eqref{eq:uniform}.  

Let $H_1:=H(X,A'_0,A'_1)$, and $H_2:=H(X,D_0,D_1)$. 
Let $\gamma_1 = \gamma(A'_1)$, and $\gamma_2 = \gamma(D_1)$. 
To show that $(X,D_0,D_1)$ is tight, it suffices to show that 
$p(X,D_0,D_1) \leq p(X,A'_0,A'_1)$.  
This boils down to showing that  
\begin{equation} \label{eq:step1}
\sum_{K\in\comp(H_2)} \floor{ \frac{|V(K) \cap (D_0 \cup D_1) |}{2} } \leq 
\sum_{K\in\comp(H_1)} \floor{ \frac{|V(K) \cap (A'_0 \cup A'_1) |}{2} }.
\end{equation}
Observe that $H_1$ and $H_2$ have the same vertex-set $V(H)-X$, and $H_1$ is a subgraph of
$H_2$. 
Every component of $H_2$ is either also a component of $H_1$, 
or is obtained by merging some components of $H_1$. Moreover, there is at 
most one component of $H_2$ that corresponds to the latter case; call this component
$K_2$, and let $\K$ be the components of $H_1$ that are merged to form $K_2$. 
Clearly, every component of $H_2$ that is also a component of $H_1$ contributes the same
to both the LHS and RHS of \eqref{eq:step1}. So we focus on component $K_2$ and show that
the contribution of $K_2$ to the LHS of \eqref{eq:step1} is at most the total contribution
of components in $\K$ to the RHS of \eqref{eq:step1}.
 
Note that all vertices in $[y]-X$ belong to the same component of $H_1$: since $y$
violates \eqref{eq:uniform}, it is not hard to see that there is a spanning subgraph of
$[y]-X$ that is not contained in $E(A'_0 \cup A'_1,\gamma_1)$. 
Every edge $e=(w,z)\in E(H_2)$ that merges two components in $\K$ is therefore such
that $\{w,z\}\not\sse[y]$, and lies in $E(A'_0\cup A'_1,\gm_1)-E(D_0\cup D_1,\gm_2)$. 
So we have $\gm_{1,e}=0$, $w,z\in A'_0\cup A'_1$, and $\{w,z\}\not\sse D_0\cup D_1$. It
follows that $e$ has exactly one end in $([y]-X)\cap (A'_0\cup A'_1)$, and the other
end not in $[y]$. 
Therefore, $|\K|\leq 1+\bigl|([y]-X)\cap(A'_0\cup A'_1)\bigr|$. 
It follows that
\begin{equation*}
\begin{split}
|V(K_2) \cap (D_0 \cup D_1)| & = \sum_{K \in \K}\bigl|V(K) \cap (A'_0 \cup A'_1)\bigr| 
- \bigl|([y]-X) \cap (A'_0 \cup A'_1)\bigr| \\ 
& \leq \sum_{K \in \K}\bigl|V(K) \cap (A'_0 \cup A'_1)\bigr| - (|\K| - 1).
\end{split}
\end{equation*}
Combining with Claim~\ref{basic}, we obtain that
$$ 
\sum_{K \in \K}\floor{\frac{|V(K) \cap (A'_0 \cup A'_1)|}{2}}\geq
\floor{\frac{\sum_{k\in\K}|V(K)\cap (A'_0\cap A'_1)-(|\K|-1)}{2}}
\geq\floor{\frac{|V(K_2)\cap (D_0\cup D_1)|}{2}}.
$$
This completes the proof that $(X,D_0,D_1)$ is a tight triple.

\paragraph{Step 2: removing vertices in \boldmath $X$.} 
Recall that Step 1 yields a tight triple $(X,\tA_0,\tA_1)$ satisfying
\eqref{eq:uniform} for every vertex $y \in V(G)$. 
We now show how to convert $(X,\tA_0,\tA_1)$ to a tight triple $(\es,B_0,B_1)$ 
satisfying \eqref{eq:emptyx} for every $x\in V(G)$.
Let $\tA_*$ denote 
$V(H(X,\tA_0,\tA_1))-\tA_0-\tA_1=V(H)-X-\tA_0-\tA_1$. 
Let $y\in V(G)$ be such that $[y]\cap X\neq\es$. Choose some vertex $z\in[y]\cap X$. 
We know that \eqref{eq:uniform} holds for $y$. Let $Y:=X-\{z\}$. 
If $\tA_0\cap[y]\neq\es$, set $J_0=\tA_0\cup\{z\}$, else set
$J_0=\tA_0$. If $\tA_1\cap[y]\neq\es$, set $J_1=\tA_1\cup\{z\}$, else set 
$J_1=\tA_1$. Notice that 
for every vertex $x\in V(G)$, all vertices in $[x]-Y$ are either in $J_0$, or in
$J_1$, or in $J^*:=V(H)-Y-J_0-J_1$.
We show that $(Y,J_0,J_1)$ is a tight triple; so by repeating this process, we will 
eventually obtain a tight triple $(\es,B_0,B_1)$ satisfying \eqref{eq:emptyx}.

Let $H_1:=H(X,\tA_0,\tA_1)$ and $H_2:=H(Y,J_0,J_1)$, so  
$H_1$ is a subgraph of $H_2$, and $V(H_2)=V(H_1)\cup\{z\}$.
Again, it suffices to show that $p(Y,J_0,J_1)\leq p(X,\tA_0,\tA_1)$.
This amounts to showing that
\begin{equation} \label{eq:step2}
\sum_{K\in\comp(H_2)} \floor{ \frac{|V(K) \cap (J_0 \cup J_1) |}{2} } \leq 
1+\sum_{K\in\comp(H_1)} \floor{ \frac{|V(K) \cap (\tA_0 \cup \tA_1) |}{2} }.
\end{equation}
Let $\gamma_2:=\gamma(J_1)$. 
Let $w$ be the other end of the non-clique edge of $H$ incident to $z$.
Let $K_1$ be the component of $H_1$ containing $w$. 
Let $K_2$ be the component of $H_2$ containing $z$. 
Observe that since $(X,\tA_0,\tA_1)$ is a tight triple, the inequality in \eqref{eq:step2}
cannot be strict (otherwise, we would have $\nu(H,[s],\gm)>p(Y,J_0,J_1)$). This implies
that $z$ cannot be an isolated vertex of $H_2$.
 
If $z\in J_0\cup J_1$, then all edges of $H_2$ belonging to the clique $[y]$ lie in
$E(J_0\cup J_1,\gm_2)$. So we must have that $V(K_2)\supseteq\{w,z\}$ (as $z$ cannot be
an isolated vertex), 
and $V(K_2)=V(K_1)\cup\{z\}$. Clearly, $\comp(H_2)-\{K_2\}=\comp(H_1)-\{K_1\}$, and
all components in this set contribute equally to the LHS and RHS of \eqref{eq:step2}, so
\eqref{eq:step2} follows in this case by noting that 
$|V(K_2)\cap(J_0\cup J_1)|=1+|V(K_1)\cap(\tA_0\cup\tA_1)|$. 

If $z\in J_*$, then $w\in V(K_2)$, so $V(K_1)\sse V(K_2)$.
All vertices of $[y]-X$ belong to the same component of $H_1$; call this component
$K'_1$. (One can argue that $K_1\neq K'_1$, but we do not need this below.) 
Clearly, $V(K'_1)\sse V(K_2)$, so we have $V(K_2)=V(K_1)\cup V(K'_1)\cup\{z\}$.
We have $\comp(H_2)-\{K_2\}=\comp(H_1)-\{K_1,K'_1\}$, and all these components contribute
equally to the LHS and RHS of \eqref{eq:step2}.
Finally, 
\begin{equation*}
\begin{split}
\floor{\frac{|V(K_2)\cap (J_0\cup J_1)|}{2}} & =
\floor{\frac{|V(K_1)\cap(\tA_0\cup\tA_1)|}{2}+\frac{|V(K'_1)\cap(\tA_0\cup\tA_1)|}{2}} \\
& \leq
1+\floor{\frac{|V(K_1)\cap(\tA_0\cup\tA_1)|}{2}}+\floor{\frac{|V(K'_1)\cap(\tA_0\cup\tA_1)|}{2}}
\end{split}
\end{equation*}
where the inequality above follows from Claim~\ref{basic}.
Thus, \eqref{eq:step2} holds in this case as well.
\end{proofof}

\section{Extensions}

\paragraph{Odd trails in signed graphs.}
A {\em signed graph} is a tuple $\bigl(G=(V,E),\Sg\bigr)$, where $G$ is undirected and
$\Sg\sse E$. A set $F$ of edges is now called odd if $|F\cap\Sg|$ is odd. 
Our results extend to the more-general
setting of packing and covering odd $(u,v)$-trails in a signed graph. In particular,
Theorems~\ref{thm:2k}, \ref{thm:sstrails} and~\ref{thm:contacts} 
{\em hold without any changes}. Theorem~\ref{thm:sstrails} follows simply because it
utilizes Theorem~\ref{grpthm}, which applies to the even more-general setting of
group-labeled graphs. Theorem~\ref{thm:contacts} holds because it uses basic parity arguments:
if we simply replace parity with parity with respect to $\Sg$ (i.e., instead of parity of
$F$, we now consider parity of $|F\cap\Sg|$), then everything goes through. Finally, as
before, combining the above two results yields (the extension of) Theorem~\ref{thm:2k}.

\paragraph{Odd $(C,D)$-trails.}
This is the generalization of the odd $(u,v)$-trails setting, where we have disjoint
sets $C,D\sse V$. 
Our results yield a factor-2 gap 
between the the minimum number of edges needed to cover all odd $(C,D)$-trails 
and the maximum number of edge-disjoint odd $(C,D)$-trails.

We achieve this as follows. First, we prove a generalization of Theorem~\ref{thm:sstrails} 
showing that for any integer $k\geq 0$, we can either obtain $k$ edge-disjoint odd
$(C\cup D,C\cup D)$-trails, or an odd-$(C\cup D,C\cup D)$-trail cover of size at most
$2k-2$. This follows by again utilizing Theorem~\ref{grpthm}: we use the same gadget graph
$H$ and group $\Gm=\grp$ as in Section~\ref{sec:sstrails}, but now take $A$ to be all the
clique nodes corresponding to the nodes of $C$ and $D$. It is not hard to see that the
same translation from the group-labeled setting applies here, and yields the above
generalization of Theorem~\ref{thm:sstrails}. 
Next, we observe that Theorem~\ref{thm:contacts} can still be applied in this more-general
setting to show that if we have a collection $\hT$ of $k$ edge-disjoint odd 
$(C\cup D,C\cup D)$-trails, and (at least) $2k$ edge-disjoint $(C,D)$-paths, then we can
obtain $k$ edge-disjoint odd $(C,D)$-trails. 
A simple way of seeing this is that we may simply contract $C$ and $D$ to
form supernodes $u$ and $v$, and then utilize the earlier proof; the resulting graph may
have loops, but this can be avoided by reworking the proof of Theorem~\ref{thm:contacts}
to work directly with $\hT$.

\bibliographystyle{plain}

\appendix

\section{Tight example for Theorem~\ref{thm:2k}} \label{append-2k}
We prove that for the graph $G=(V,E)$ shown in Fig.~\ref{2gap}, we have $\nu(u,v)=k$, and 
$\tau(u,v)=2k+1$. 
For any $i=1,\ldots,k$, let $B_i$ denote the subgraph of $G$ induced by the vertices
$\{u,a_i,b_i,\ldots,g_i,h_i,v\}$. 

The size of a minimum $(u,v)$-cut in $G_k$ is $2k+1$, so $\tau(u,v)\leq 2k+1$. 
It is easy to see by inspection that any cover $Z$ of $G$ must contain at least $2$ edges
from each $B_i$, as otherwise $B_i-Z$ will have an odd $(u,v)$-path. Suppose, for a
contradiction, that $|Z|\leq 2k$. Then $G-Z$ contains a $(u,v)$-path $P$, and $Z$ contains
exactly $2$ edges from each $B_i$. We argue that $G-Z$ contains an odd $(u,v)$-trail,
yielding a contradiction and thereby showing that $\tau(u,v)=2k+1$.
$Z$ must contain exactly one edge from each of
the two edge-disjoint triangles $\{u,a_1,b_1\}$ and $\{v,g_1,h_1\}$ in $B_1$ containing
$u$ and $v$ respectively, as otherwise, one can combine $P$ and such a triangle to obtain
an odd $(u,v)$-trail in $G-Z$. But then $B_1-Z$ is connected, and contains all edges
incident to the nodes $c_1,d_1,e_1,f_1$, which implies that $B_1-Z$ contains an odd
$(u,v)$-path.   

\medskip

We now argue that $\nu(u,v)=k$. We have $\nu(u,v)\geq k$, as each $B_i$ contains an odd  
$(u,v)$-path. 
Among all collections of $\nu(u,v)$ edge-disjoint odd $(u,v)$-trails, let $\T$ be one with
the fewest number of edges. 
Let $P_1,\ldots,P_r$ be the odd $(u,v)$-paths in $\T$, and $T_1,\ldots,T_\ell$ be the odd 
$(u,v)$-trails in $\T$ that are not $(u,v)$-paths.
Observe that each $(u,v)$-path $P\in\{P_1,\ldots,P_r\}$ must be contained in some subgraph 
$B_i$, and $B_i-P$ is disconnected. Next, observe that every odd $(u,v)$-trail 
$T\in\{T_1,\ldots,T_\ell\}$ consists of an even $(u,v)$-path concatenated with an odd
$(u,u)$-trail or an odd $(v,v)$-trail, say $Z$. Moreover, $Z$ must in fact be a triangle
containing $u$ or $v$ due to the minimality of $\T$. (Suppose $Z$ is an odd
$(u,u)$-trail. Then $Z$ must contain an odd cycle contained in some subgraph $B_i$ and the
edges in $B_i$ incident to $u$; replacing $Z$ in the trail $T$ with the triangle in $B_i$
incident to $u$ yields an odd $(u,v)$-trail $T'$ with $|E(T')|<|E(T)|$ that is
edge-disjoint from $\T-\{T\}$, contradicting the minimality of $\T$. The case where $Z$ is
an odd $(v,v)$-trail is completely analogous.)  
Thus, every trail in $\{T_1,\ldots,T_\ell\}$ consists of an even $(u,v)$-path and a
triangle incident to $u$ or $v$. Without loss of generality, let $\ceil{\frac{\ell}{2}}$
of these trails contain triangles incident to $u$. 
Let $\Zc_u$ be the collection of triangles contained in the trails $\{T_1,\ldots,T_\ell\}$
that are incident to $u$, so $|\Zc_u|\geq\ceil{\frac{\ell}{2}}$.   

Consider the graph $G'=G-\bigcup_{i=1}^rP_i-\bigcup_{Z\in\Zc_u}Z$. The $(u,v)$ connectivity
in $G'$ is at most $2k+1-2r-2|\Zc_u|$: deleting each $P_i$ leaves some $B_{i'}$
disconnected, and so deleting either a $P_i$ or a triangle in $\Zc_u$ decreases the
$(u,v)$-connectivity by 2. Also, $G'$ contains $\ell$ (even) $(u,v)$-paths. So we have
$\ell\leq 2k+1-2r-2\ceil{\frac{\ell}{2}}$, which implies that $\ell+r\leq\frac{2k+1}{2}$,
and hence that $\ell+r\leq k$. Thus, $\nu(u,v)=k$.

\section{Tight example for Theorem~\ref{thm:contacts}} \label{append-contacts}
We now provide an example showing that Theorem~\ref{thm:contacts} is tight; in particular,
the requirement that $\ld(u,v)\geq 2|\hT|$ cannot be weakened.
Consider the graph $G$ shown in Figure~\ref{fig:2kconn}. 
$\T = \{ u x_i y_i v \}_{i=1}^{k-1} \cup \{ u z_1 z_2 u \} \cup \{ v z_5 z_6 v \}$ is a
collection of $k+1$ edge-disjoint odd $(\uv,\uv)$-trails in $G$.
Observe that $\lambda(u,v) = 2k + 1=2|\T|-1$.
But $G$ contains only $k$ edge-disjoint odd $(u,v)$-trails: the subgraph induced by $\uv$, 
the $z_i$s and $w$ contains only one odd $(u,v)$-trail.

\begin{figure}[h!]
\centering
\begin{tikzpicture}[scale=0.7]

\setcounter{xx}{0}
\setcounter{yy}{0}
\coordinate (centre) at (\value{xx},\value{yy});
\node (dummy) at (centre) {};

\coordinate (u) at ([xshift=-3cm]centre);
\coordinate (v) at ([xshift=3cm]centre);
\coordinate (xa) at ([yshift=5cm]centre);
\coordinate (ya) at ([yshift=-1.5cm]xa);
\coordinate (xb) at ([yshift=-1.5cm]ya);
\coordinate (yb) at ([yshift=-1.5cm]xb);
\coordinate (za) at ([xshift=1.5cm,yshift=-2.5cm]u);
\coordinate (zb) at ([yshift=-1.5cm]za);
\coordinate (zc) at ([xshift=1.5cm]za);
\coordinate (zd) at ([xshift=1.5cm]zb);
\coordinate (ze) at ([xshift=1.5cm]zc);
\coordinate (zf) at ([xshift=1.5cm]zd);
\coordinate (w) at ([yshift=-1.5cm]yb);

\draw [fill=black] (u) circle (\rad) node[left] {$u$};
\draw [fill=black] (v) circle (\rad) node[right] {$v$};
\draw [fill=black] (xa) circle (\rad) node[above] {$x_1$};
\draw [fill=black] (ya) circle (\rad) node[right] {$y_1$};
\draw [fill=black] (xb) circle (\rad) node[above] {$x_{k-1}$};
\draw [fill=black] (yb) circle (\rad) node[below] {$y_{k-1}$};
\draw [fill=black] (za) circle (\rad) node[above right] {$z_1$};
\draw [fill=black] (zb) circle (\rad) node[below] {$z_2$};
\draw [fill=black] (zc) circle (\rad) node[above] {$z_3$};
\draw [fill=black] (zd) circle (\rad) node[below] {$z_4$};
\draw [fill=black] (ze) circle (\rad) node[above left] {$z_5$};
\draw [fill=black] (zf) circle (\rad) node[below] {$z_6$};
\draw [fill=black] (w) circle (\rad) node[below] {$w$};

\draw (u) -- (za);
\draw (u) -- (zb);

\draw (v) -- (ze);
\draw (v) -- (zf);

\draw (za) -- (zc);
\draw (zb) -- (zd);

\draw (zc) -- (ze);
\draw (zd) -- (zf);

\draw (u) -- (w);
\draw (v) -- (w);

\draw (u) -- (xa);
\draw (u) -- (xb);
\draw (u) -- (ya);
\draw (u) -- (yb);

\draw (v) -- (xa);
\draw (v) -- (xb);
\draw (v) -- (ya);
\draw (v) -- (yb);

\draw (xa) -- (ya);
\draw (xb) -- (yb);
\draw (za) -- (zb);
\draw (ze) -- (zf);

\node[font=\bf] at ([yshift=-0.75cm]ya) {\vdots};

\end{tikzpicture}

\caption{Example showing that Theorem~\ref{thm:contacts} is tight.} \label{fig:2kconn} 
\end{figure}
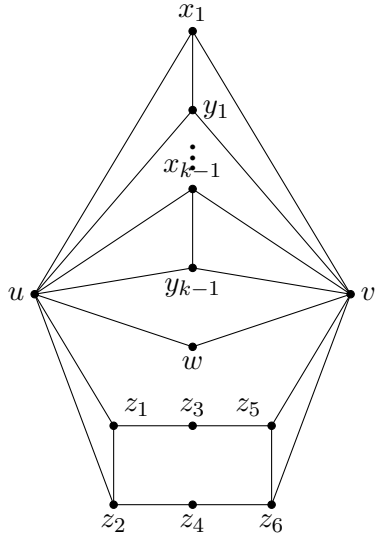

\section{Lower bound for covers constructed in~\cite{ChurchleyMW16}} \label{append-lbound}
Churchley et al.~\cite{ChurchleyMW16} construct odd-$(u,v)$-trail covers of the form 
$\dt(X)$, where $X$ is a $u$-$v$ cut, or of the form
$\dt(X)\cup(E(X)\sm F)$, where $\uv\sse X$, and $(X,E(X)\cap F)$ is bipartite with $u$ and $v$
on the same side of the bipartition. The following is a family $\{H_k\}_{k\geq 1}$ of
graphs where, for $H_k$ we have $\nu(u,v)=k$, $\tau(u,v)=2k$, but any cover of the above
form has size at least $3k$. This shows that covers of the above form cannot yield a bound
better than $3$ on the covering-vs-packing ratio for odd $(u,v)$ trails. (The graph $H_k$
has parallel edges; if desired, the parallel edges can be eliminated by replacing each
parallel edge with a length-$3$ path (to preserve trail parities).)

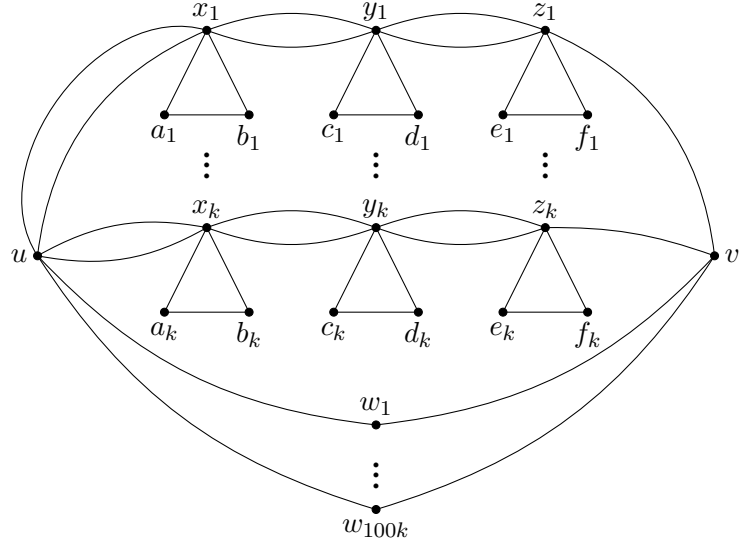
\begin{figure}[h!]
\centering
\begin{tikzpicture}[scale=0.75]

\setcounter{xx}{0}
\setcounter{yy}{0}
\coordinate (centre) at (\value{xx},\value{yy});
\node (dummy) at (centre) {};

\coordinate (u) at ([xshift=-6cm]centre);
\coordinate (v) at ([xshift=6cm]centre);

\coordinate (xa) at ([xshift=3cm,yshift=4cm]u);
\coordinate (ya) at ([xshift=3cm]xa);
\coordinate (za) at ([xshift=3cm]ya);

\coordinate (aa) at ([xshift=-0.75cm,yshift=-1.5cm]xa);
\coordinate (ba) at ([xshift=0.75cm,yshift=-1.5cm]xa);
\coordinate (ca) at ([xshift=-0.75cm,yshift=-1.5cm]ya);
\coordinate (da) at ([xshift=0.75cm,yshift=-1.5cm]ya);
\coordinate (ea) at ([xshift=-0.75cm,yshift=-1.5cm]za);
\coordinate (fa) at ([xshift=0.75cm,yshift=-1.5cm]za);

\draw [fill=black] (u) circle (\rad) node[left] {$u$};
\draw [fill=black] (v) circle (\rad) node[right] {$v$};
\draw [fill=black] (xa) circle (\rad) node[above] {$x_1$};
\draw [fill=black] (ya) circle (\rad) node[above] {$y_1$};
\draw [fill=black] (za) circle (\rad) node[above] {$z_1$};

\draw [fill=black] (aa) circle (\rad) node[below] {$a_1$};
\draw [fill=black] (ba) circle (\rad) node[below] {$b_1$};
\draw [fill=black] (ca) circle (\rad) node[below] {$c_1$};
\draw [fill=black] (da) circle (\rad) node[below] {$d_1$};
\draw [fill=black] (ea) circle (\rad) node[below] {$e_1$};
\draw [fill=black] (fa) circle (\rad) node[below] {$f_1$};

\draw (u) edge[bend left=70] (xa);
\draw (u) edge[bend left=30] (xa);

\draw (xa) edge[bend left=20] (ya);
\draw (xa) edge[bend right=20] (ya);

\draw (ya) edge[bend left=20] (za);
\draw (ya) edge[bend right=20] (za);

\draw (za) edge[bend left=30] (v);

\draw (xa) -- (aa);
\draw (xa) -- (ba);
\draw (ya) -- (ca);
\draw (ya) -- (da);
\draw (za) -- (ea);
\draw (za) -- (fa);
\draw (aa) -- (ba);
\draw (ca) -- (da);
\draw (ea) -- (fa);

\node[font=\bf] at ([yshift=-2.25cm]xa) {\vdots};
\node[font=\bf] at ([yshift=-2.25cm]ya) {\vdots};
\node[font=\bf] at ([yshift=-2.25cm]za) {\vdots};

\coordinate (xb) at ([yshift=-3.5cm]xa);
\coordinate (yb) at ([yshift=-3.5cm]ya);
\coordinate (zb) at ([yshift=-3.5cm]za);

\coordinate (ab) at ([xshift=-0.75cm,yshift=-1.5cm]xb);
\coordinate (bb) at ([xshift=0.75cm,yshift=-1.5cm]xb);
\coordinate (cb) at ([xshift=-0.75cm,yshift=-1.5cm]yb);
\coordinate (db) at ([xshift=0.75cm,yshift=-1.5cm]yb);
\coordinate (eb) at ([xshift=-0.75cm,yshift=-1.5cm]zb);
\coordinate (fb) at ([xshift=0.75cm,yshift=-1.5cm]zb);

\draw [fill=black] (xb) circle (\rad) node[above] {$x_k$};
\draw [fill=black] (yb) circle (\rad) node[above] {$y_k$};
\draw [fill=black] (zb) circle (\rad) node[above] {$z_k$};

\draw [fill=black] (ab) circle (\rad) node[below] {$a_k$};
\draw [fill=black] (bb) circle (\rad) node[below] {$b_k$};
\draw [fill=black] (cb) circle (\rad) node[below] {$c_k$};
\draw [fill=black] (db) circle (\rad) node[below] {$d_k$};
\draw [fill=black] (eb) circle (\rad) node[below] {$e_k$};
\draw [fill=black] (fb) circle (\rad) node[below] {$f_k$};

\draw (u) edge[bend left=20] (xb);
\draw (u) edge[bend right=20] (xb);

\draw (xb) edge[bend left=20] (yb);
\draw (xb) edge[bend right=20] (yb);

\draw (yb) edge[bend left=20] (zb);
\draw (yb) edge[bend right=20] (zb);

\draw (zb) edge[bend left=10] (v);

\draw (xb) -- (ab);
\draw (xb) -- (bb);
\draw (yb) -- (cb);
\draw (yb) -- (db);
\draw (zb) -- (eb);
\draw (zb) -- (fb);
\draw (ab) -- (bb);
\draw (cb) -- (db);
\draw (eb) -- (fb);

\coordinate (wa) at ([yshift=-3cm]centre);
\coordinate (wb) at ([yshift=-1.5cm]wa);
\draw [fill=black] (wa) circle (\rad) node[above] {$w_1$};
\draw [fill=black] (wb) circle (\rad) node[below] {$w_{100k}$};

\draw (u) edge[bend right=20] (wa);
\draw (u) edge[bend right=20] (wb);
\draw (v) edge[bend left=20] (wa);
\draw (v) edge[bend left=20] (wb);

\node[font=\bf] at ([yshift=-0.75cm]wa) {\vdots};

\end{tikzpicture}

\caption{Graph $H_k$ with $\nu(u,v)=k$, $\tau(u,v)=2k$.
We have $|\delta(X)| \cup |E(X) \sm F| \geq 3k$ for any $X, F$ such that $\uv\sse X$
and $(X,E(X)\cap F)$ is bipartite with $u,v$ on the same side.} 
\label{fig:3k} 
\end{figure}

Let $C_i$ denote the vertex-set $\{u,v,x_i,y_i,z_i,a_i,b_i,c_i,d_i,e_i,f_i\}$. 
Observe that every odd trail $T$ in $G$ uses an edge from at least one of the odd cycles  
in $\{x_i a_i b_i x_i \}_{i=1}^{k} \cup \{y_i c_i d_i y_i \}_{i=1}^{k} \cup \{z_i e_i f_i z_i
\}_{i=1}^{k}$. 
For each $1 \leq i \leq k$, removing the two parallel edges between $u$ and $x_i$
along with the edge $v z_i$ separates $C_i\sm\{u,v\}$ from $\{u,v\}$, so at most one trail 
can use the nodes of $C_i\sm\{u,v\}$.
It follows that the packing number is at most $k$. 
It is trivial to find $k$ edge-disjoint odd $(u,v)$-trails in $G$ and hence
$\nu(u,v)=k$. Next, the covering number is at most $2k$ since the edge-set consisting of
all the edges (including parallel copies) between $u$ and $x_i$ for all $1 \leq i \leq k$
is an odd $(u,v)$-trail cover. Deleting at most one edge from the induced
subgraph $G[C_i]$ does not destroy all the odd $(u,v)$-trails using at least one odd cycle
from $G[C_i]$ hence we need at least $2k$ edges in any cover. Thus, $\tau(u,v)=2k$.  

Since $\ld(u,v)\geq 100k$, every $u$-$v$ cut $X$ has $|\dt(X)|\geq 100k$.
Finally, we verify that $|\delta(X)| \cup |E(X) \sm F| \geq 3k$ for any $X$ ,$F$
such that $\uv\sse X$ and $(X,E(X)\cap F)$ is bipartite with $u,v$ on the same side. For
this it is convenient to consider the graph shown in Fig.~\ref{fig:3ks}, which is the
graph in Fig.~\ref{fig:3k}, with $u$ and $v$ identified to form $s$.

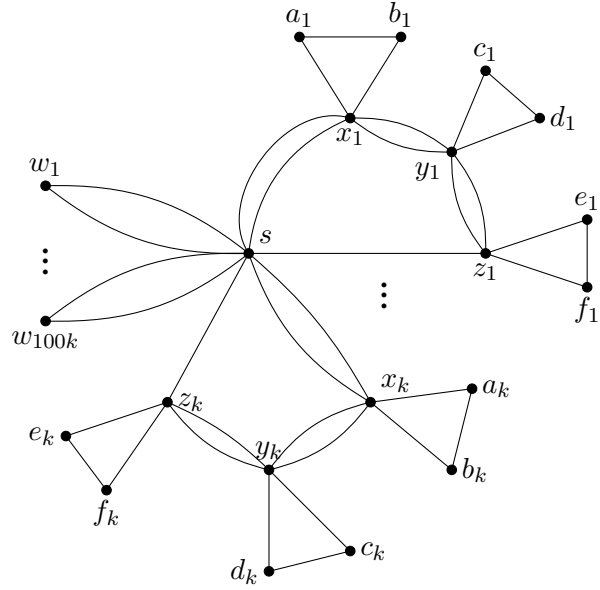
\begin{figure}[h!]
\centering
\begin{tikzpicture}[scale=0.9]

\setcounter{xx}{0}
\setcounter{yy}{0}
\coordinate (centre) at (\value{xx},\value{yy});
\node (dummy) at (centre) {};

\coordinate (s) at (centre);

\coordinate (xa) at ([xshift=1.5cm,yshift=2cm]s);
\coordinate (ya) at ([xshift=1.5cm,yshift=-0.5cm]xa);
\coordinate (za) at ([xshift=0.5cm,yshift=-1.5cm]ya);

\coordinate (aa) at ([xshift=-0.75cm,yshift=1.2cm]xa);
\coordinate (ba) at ([xshift=0.75cm,yshift=1.2cm]xa);
\coordinate (ca) at ([xshift=0.5cm,yshift=1.2cm]ya);
\coordinate (da) at ([xshift=1.3cm,yshift=0.5cm]ya);
\coordinate (ea) at ([xshift=1.5cm,yshift=0.5cm]za);
\coordinate (fa) at ([xshift=1.5cm,yshift=-0.5cm]za);

\draw [fill=black] (s) circle (\rad) node[above right] {$s$};
\draw [fill=black] (xa) circle (\rad) node[below] {$x_1$};
\draw [fill=black] (ya) circle (\rad) node[below left] {$y_1$};
\draw [fill=black] (za) circle (\rad) node[below] {$z_1$};

\draw [fill=black] (aa) circle (\rad) node[above] {$a_1$};
\draw [fill=black] (ba) circle (\rad) node[above] {$b_1$};
\draw [fill=black] (ca) circle (\rad) node[above] {$c_1$};
\draw [fill=black] (da) circle (\rad) node[right] {$d_1$};
\draw [fill=black] (ea) circle (\rad) node[above] {$e_1$};
\draw [fill=black] (fa) circle (\rad) node[below] {$f_1$};

\draw (s) edge[bend left=70] (xa);
\draw (s) edge[bend left=30] (xa);

\draw (xa) edge[bend left=20] (ya);
\draw (xa) edge[bend right=20] (ya);

\draw (ya) edge[bend left=20] (za);
\draw (ya) edge[bend right=20] (za);

\draw (s) edge[bend left=0] (za);

\draw (aa) -- (ba);

\draw (xa) -- (aa);
\draw (xa) -- (ba);
\draw (ya) -- (ca);
\draw (ya) -- (da);
\draw (za) -- (ea);
\draw (za) -- (fa);

\draw (ca) -- (da);
\draw (ea) -- (fa);

\node[font=\bf] at ([xshift=0.5cm,yshift=-2.5cm]xa) {\vdots};

\coordinate (xb) at ([xshift=1.8cm,yshift=-2.2cm]s);
\coordinate (yb) at ([xshift=-1.5cm,yshift=-1cm]xb);
\coordinate (zb) at ([xshift=-1.5cm,yshift=1cm]yb);
\coordinate (ab) at ([xshift=1.5cm,yshift=0.2cm]xb);
\coordinate (bb) at ([xshift=1.2cm,yshift=-1cm]xb);
\coordinate (cb) at ([xshift=1.2cm,yshift=-1.2cm]yb);
\coordinate (db) at ([xshift=0cm,yshift=-1.5cm]yb);
\coordinate (eb) at ([xshift=-1.5cm,yshift=-0.5cm]zb);
\coordinate (fb) at ([xshift=-0.9cm,yshift=-1.3cm]zb);
\draw [fill=black] (xb) circle (\rad) node[above right] {$x_k$};
\draw [fill=black] (yb) circle (\rad) node[above] {$y_k$};
\draw [fill=black] (zb) circle (\rad) node[right] {$z_k$};

\draw [fill=black] (ab) circle (\rad) node[right] {$a_k$};
\draw [fill=black] (bb) circle (\rad) node[right] {$b_k$};
\draw [fill=black] (cb) circle (\rad) node[right] {$c_k$};
\draw [fill=black] (db) circle (\rad) node[left] {$d_k$};
\draw [fill=black] (eb) circle (\rad) node[left] {$e_k$};
\draw [fill=black] (fb) circle (\rad) node[below] {$f_k$};

\draw (s) edge[bend left=10] (xb);
\draw (s) edge[bend right=20] (xb);

\draw (xb) edge[bend left=20] (yb);
\draw (xb) edge[bend right=20] (yb);

\draw (yb) edge[bend left=20] (zb);
\draw (yb) edge[bend right=10] (zb);

\draw (zb) edge[bend left=0] (s);

\draw (xb) -- (ab);
\draw (xb) -- (bb);
\draw (ab) -- (bb);

\draw (yb) -- (cb);
\draw (yb) -- (db);
\draw (cb) -- (db);

\draw (zb) -- (eb);
\draw (zb) -- (fb);
\draw (eb) -- (fb);

\coordinate (wa) at ([xshift=-3cm,yshift=1cm]s);
\coordinate (wb) at ([xshift=-3cm,yshift=-1cm]s);
\draw [fill=black] (wa) circle (\rad) node[above] {$w_1$};
\draw [fill=black] (wb) circle (\rad) node[below] {$w_{100k}$};

\draw (s) edge[bend right=20] (wa);
\draw (s) edge[bend right=20] (wb);
\draw (s) edge[bend left=20] (wa);
\draw (s) edge[bend left=20] (wb);

\node[font=\bf] at ([yshift=-1cm]wa) {\vdots};

\end{tikzpicture}

\caption{Graph from Figure~\ref{fig:3k} with vertices $u$ and $v$ identified to form $s$.} \label{fig:3ks}
\end{figure}

Let $C'_i=\{s,x_i,y_i,z_i,a_i,b_i,c_i,d_i,e_i,f_i\}$. 
For any node-set $X$ with $s\in X$, we have 
$|\dt(X)|\geq\sum_{i=1}^k|E(C'_i\cap X,C'_i\sm X)|$. 
For $F$ such that $(X,E(X)\cap F)$ is bipartite, we lower bound $|E(X)\sm F|$ by counting
the number of triangles (i.e., $\{x_i,a_i,b_i\}$, $\{y_i,c_i,d_i\}$, or $\{z_i,e_i,f_i\}$
for some $i$) in $E(X)$. This contribution is decoupled across the $C'_i$s, so overall each
$C'_i$ contributes $|E(C'_i\cap X,C'_i\sm X)|+\text{(number of triangles in $C'_i\cap X$)}$. 
We argue that this contribution is at least $3$ for each $C'_i$.
Each $C'_i$ is 2-edge-connected, and the cuts of size 2 in $C'_i$ are those that contain
two edges of one of the $C'_i$-triangles. So if $|E(C'_i\cap X, C'_i\sm X)|\leq 2$, then
$C'_i\cap X$ contains at least one triangle, and the contribution is at least $3$.

\section{Tight example for Theorem~\ref{thm:sstrails}} \label{append-sstexmpl}
The example is the same as the graph $G$ in Fig.~\ref{fig:3ks}, which, recall is
obtained from the graph $G'$ in Fig.~\ref{fig:3k} by identifying $u$ and $v$. Since
$\tau(u,v;G')=2k$, we have $\tau(s,s;G)\geq 2k$. We argue that $\nu(s,s,G)\leq k$, showing
that the two conclusions of the theorem are tight as seen by the inputs $k$ and $k+1$.

The upper bound on $\nu(s,s,G)$ follows from the same reasoning as before.
Recall that we have $C'_i=\{s,x_i,y_i,z_i,a_i,b_i,c_i,d_i,e_i,f_i\}$.
Every odd trail $T$ in $G$ must use an edge in at least one of the triangles. For each 
$1 \leq i \leq k$, the two parallel edges between $s$ and $x_i$ and $sz_i$ separate the
triangle vertices of $C'_i$ from $s$, so each $C'_i$ contributes at most 1 to
$\nu(s,s,G)$. Therefore, $\nu(s,s,G)\leq k$.

\end{document}